\documentclass[onecolumn, conference]{IEEEtran}

\usepackage{amssymb}
\usepackage{amsmath}
\usepackage{epic,eepic}
\usepackage{epsfig}
\usepackage{cite}
\usepackage{verbatim}
\usepackage{enumerate}
\usepackage{subfigure}
\usepackage{multicol}
\usepackage{picinpar}
\usepackage{pstricks}
\usepackage{pst-node}

\newcommand{\define}{\stackrel{\triangle}{=}}
\newcommand{\spa}{\mbox{span}}

\newtheorem{theorem}{\bf Theorem}

\newcommand{\xH}{\mathbf{H}}
\newcommand{\xX}{\mathbf{X}}
\newcommand{\xY}{\mathbf{Y}}
\newcommand{\xZ}{\mathbf{Z}}
\newcommand{\xV}{\mathbf{V}}
\newcommand{\xe}{\mathbf{e}}
\newcommand{\xw}{\mathbf{w}}

\newcommand{\barxV}{\mathbf{\bar{V}}}
\newcommand{\barxH}{\mathbf{\bar{H}}}
\newcommand{\barxX}{\mathbf{\bar{X}}}
\newcommand{\barxY}{\mathbf{\bar{Y}}}
\newcommand{\barxZ}{\mathbf{\bar{Z}}}

\begin{document}
\setcounter{page}{1}
\title{{Interference Alignment and Spatial Degrees of Freedom for the $K$ User Interference Channel}}
\author{\authorblockN{Viveck R. Cadambe, Syed A. Jafar}
\authorblockA{Electrical Engineering and Computer Science\\
University of California Irvine, \\
Irvine, California, 92697, USA\\
Email: {vcadambe@uci.edu, syed@uci.edu}\\ \vspace{-1cm}}}

\maketitle
\thispagestyle{empty}
\begin{abstract} 
While the best known outerbound for the $K$ user interference channel states that there cannot be more than $K/2$ degrees of freedom, it has been conjectured that in general the constant interference channel with any number of users has only one degree of freedom. In this paper, we explore the spatial degrees of freedom per orthogonal time and frequency dimension for the $K$ user wireless interference channel where the channel coefficients take distinct values across frequency slots but are fixed in time. We answer five closely related questions. First, we show that $K/2$ degrees of freedom can be achieved by channel design, i.e. if the nodes are allowed to choose the best constant, finite and nonzero channel coefficient values. Second, we show that if channel coefficients can not be controlled by the nodes but are selected by nature, i.e., randomly drawn from a continuous distribution, the total number of spatial degrees of freedom for the $K$ user interference channel is  almost surely $K/2$ per orthogonal time and frequency dimension. Thus, only half the spatial degrees of freedom are lost due to distributed processing of transmitted and received signals on the interference channel. Third, we show that interference alignment and zero forcing suffice to achieve all the degrees of freedom in all cases. Fourth, we show that the degrees of freedom $D$ directly lead to an $\mathcal{O}(1)$ capacity characterization of the form $C(SNR)=D\log(1+SNR)+\mathcal{O}(1)$ for the multiple access channel, the broadcast channel, the $2$ user interference channel, the $2$ user MIMO $X$ channel and the $3$ user interference channel with $M>1$ antennas at each node. It is not known if this relationship is true for all networks in general, and the $K$ user interference channel with a single antenna at all nodes in particular. Fifth, we consider the degree of freedom benefits from cognitive sharing of messages on the $3$ user interference channel.  If only one of the three messages is made available non-causally to all the nodes except its intended receiver the degrees of freedom are not increased. However, if two messages are shared among all nodes (except their intended receivers) then there are two degrees of freedom. We find that unlike the $2$ user interference channel, on the $3$ user interference channel a cognitive transmitter is not equivalent to a cognitive receiver from a degrees of freedom perspective. If one receiver has cognitive knowledge of all the other users' messages the degrees of freedom are the same as without cognitive message sharing. However, if one transmitter has cognitive knowledge of all the other users' messages then the degrees of freedom are increased from $3/2$ to $2$. 
\end{abstract}

\newpage

\section{Introduction}
The capacity of ad-hoc wireless networks is the much sought afer ``holy-grail" of network information theory \cite{Toumpis_tutorial}. While capacity characterizations have been found for centralized networks (Gaussian multiple access and broadcast networks with multiple antennas), similar capacity characterizations for most distributed communication scenarios (e.g. interference networks) remain long standing open problems. In the absence of precise capacity characterizations, researchers have pursued asymptotic and/or approximate capacity characterizations. Recent work has found the asymptotic scaling laws of network capacity as the number of nodes increases in a large network \cite{Gupta_Kumar_achievable,Ozgur_Leveque_Tse}. However, very little is known about the capacity region of smaller (finite) decentralized networks.  An important step in this direction is the recent approximate characterization of the capacity region of the  $2$ user interference channel that is accurate within one bit of the true capacity region\cite{Etkin_Tse_Wang}. Approximate characterizations of capacity regions would also be invaluable for most open problems in network information theory and may be the key to improving our understanding of wireless networks.

It can be argued that the most preliminary form of capacity characterization for a network is to characterize its degrees of freedom. The degrees of freedom represent the rate of growth of network capacity with the $\log$ of the signal to noise ratio (SNR). In most cases, the spatial degrees of freedom turn out to be the number of non-intefering paths that can be created in a wireless network through signal processing at the transmitters and receivers. While time, frequency and space all offer degrees of freedom in the form of orthogonal dimensions over which communication can take place, spatial degrees of freedom are especially interesting in a distributed network. Potentially a wireless network may have as many spatial  dimensions as the number of transmitting and receiving antennas. However, the ability to access and resolve spatial dimensions  is limited by the distributed nature of the network. Therefore, characterizing the degrees of freedom for distributed wireless networks is by itself a non-trivial problem. For example, consider an interference network with $n$ single-antenna transmitters and $n$ single-antenna receivers where each transmitter has a message for its corresponding receiver. For $n=2$ it is known that this interference network has only 1 degree of freedom \cite{MadsenIT,Jafar_dof_int}. There are no known results to show that more than $1$ degrees of freedom are achievable on the interference channel with any number of users. It is conjectured in \cite{Nosratinia-Madsen} that the $K$ user interference channel has only $1$ degree of freedom. Yet, the best known outerbound for the number of degrees of freedom with $K$ interfering nodes is $K/2$, also presented in \cite{Nosratinia-Madsen}. The unresolved gap between the inner and outerbounds highlights our lack of understanding of the capacity of wireless networks because even the number of degrees of freedom, which is the most basic characterization of the network capacity, remains an open problem. It is this open problem that we pursue in this paper. 

To gain a better understanding of the interference channel, we first consider the possibility that the transmitters and receivers can place themselves optimally, i.e., the nodes can choose their channels. Thus, the first objective of this paper is to answer the question:

{\bf Question 1: }\emph{What is the maximum number of degrees of freedom for the $K$ user interference channel if we are allowed to choose the best (finite and non-zero) channel coefficient values ?}

While the scenario above offers new insights, in practice it is more common that the channel coefficients are chosen by nature. The nodes control their coding schemes, i.e. the transmitted symbols, but not the channel coefficients, which may be assumed to be randomly drawn from a continuous distribution and causally known to all the nodes. In this context we ask the main question of this paper:

{\bf Question 2: }\emph{What is the number of degrees of freedom for the $K$ user interference channel per orthogonal time and frequency dimension ?} 

Note that the normalization by the number of orthogonal time and frequency dimensions is necessary because we wish to characterize the spatial degrees of freedom.

Spatial degrees of freedom have been characterized for several multiuser communication scenarios with multiple antenna nodes. The $(M,N)$ point to point MIMO channel has $\min(M,N)$ degrees of freedom \cite{Foschini_Gans,Telatar}, the $(M_1,M_2, N)$ multiple access channel has $\min(M_1+M_2,N)$ degrees of freedom \cite{Tse_Viswanath_Zheng}, the $(M, N_1, N_2)$ broadcast channel has $\min(M, N_1+N_2)$ degrees of freedom \cite{Yu_Cioffi,Viswanath_Tse_BC,Vishwanath_Jindal_Goldsmith}, and the $(M_1,M_2,N_1,N_2)$ interference channel has $\min(M_1+M_2,N_1+N_2,\max(M_1,N_2),\max(M_2,N_1))$ degrees of freedom \cite{Jafar_dof_int}, where $M_i$ (or $M$ when only one transmitter is present) and $N_i$ (or $N$ when only one receiver is present) indicate the number of antennas at the $i^{th}$ transmitter and receiver, respectively. 

If one tries to extrapolate these results into an understanding of the degrees of freedom for fully connected (all channel coefficients are non-zero) wireless networks with a finite number of nodes, one could arrive at the following (incorrect) intuitive inferences:
\begin{itemize}
\item \emph{The number of degrees of freedom for a wireless network with perfect channel knowledge at all nodes is an integer}.
\item \emph{The degrees of freedom of a wireless network with a finite number of nodes is not higher than the maximum number of co-located antennas at any node.}
\end{itemize}
The degrees of freedom characterizations for the point to point, multiple access, broadcast and interference scenarios described above are all consistent with both these statements. Note that the results of \cite{Boelcskei_Nabar_Oyman_Paulraj} indicate that even with single antenna sources, destinations and relay nodes  the network can have more than one degree of freedom. However, for this distributed orthogonalization result it is assumed that the number of relay nodes approaches infinity. Thus it does not contradict the intuition above which is for finite networks. Multihop networks with half-duplex relay nodes may also lead to fractional degrees of freedom due to the normalization associated with the half-duplex constraint. This is typically because of the absence of a direct link across hops, i.e. some channel coefficients are zero. For multihop networks with orthogonal hops \cite{Borade_Zheng_Gallager} has shown that the full $N$ degrees of freedom are achievable even if each intermediate hop consists of $N$ (distributed) single antenna relay nodes as long as the initial source node and the final destination nodes are equipped with $N$ antennas each. Note that the result of \cite{Borade_Zheng_Gallager} is also consistent with the inferences described above. Also, we note that channels with specialized structures or cooperation among nodes may be able to achieve higher degrees of freedom than channels whose coefficients are randomly selected from continuous distributions \cite{Lapidoth_Shamai_Wigger_IN}.

Perhaps biased by these results, most work on degrees of freedom for wireless networks has focused on either networks where some nodes are equipped with multiple antennas \cite{Jafar_dof_int, Borade_Zheng_Gallager} or networks with single antenna nodes where some form of cooperation opens up the possibility that the single antenna nodes may be able to achieve MIMO behavior \cite{MadsenIT, Nosratinia-Madsen, Devroye_Sharif}. 
Networks of single antenna nodes with no cooperation between the transmitters or receivers could be considered uninteresting from the degrees of freedom perspective as the above mentioned intuitive statements would suggest that these networks could only have 1 degree of freedom. In other words, one might argue that with a single antenna at each node it is impossible to avoid interference and therefore it is impossible to create multiple non-interfering paths necessary for degrees of freedom. 
The $2$ user interference network with a single antenna at each node is a good example of a network which adheres to all the above intuitive inferences, where indeed it can be rigorously shown that there is only one degree of freedom. Studying a $K$ user interference channel where all channel coefficients are equal will also lead to only one degree of freedom, as will the $K$ user interference channel with i.i.d. channel coefficients and no knowledge of channel coefficients at the transmitters \cite{Boelcskei_Nabar_Oyman_Paulraj}. Similarly, if all receivers observe signals that are degraded versions of, say, receiver 1's signal then it can again be argued that the MAC sum capacity when receiver 1 decodes all messages is an outerbound to the interference channel sum capacity (Carleial's outerbound \cite{Carleial_int}). Thus the degrees of freedom cannot be more than the number of antennas at receiver $1$. Finally, the conjecture that the $K$ user interference channel has only $1$ degree of freedom is also consistent with this intuition \cite{Nosratinia-Madsen}.

Clear evidence that the intuitive conclusions mentioned above do not apply to \emph{all} wireless networks is provided by the recent degree of freedom region characterization for the 2 user $X$ channel in \cite{MMK,MMKreport1, MMKreport2, Jafar_Shamai}. The 2-user $X$ channel is identical to the 2-user interference channel with the exception that each transmitter in the $X$ channel has an independent message for each receiver. Thus, unlike the interference channel which has only 2 messages, the $X$ channel has $4$ messages to be communicated between two transmitters and two receivers. Surprisingly, it was shown in \cite{Jafar_Shamai} that the $X$ channel, with only a single antenna at all nodes has $4/3$ degrees of freedom per orthogonal time/frequency dimension if the channels are time/frequency selective. This is interesting for several reasons. First, it shows that the degrees of freedom can take non-integer values. Second, it shows that the degrees of freedom of a distributed wireless network can be higher than the maximum number of co-located antennas at any node in the network. Finally, the achievability proof for the non-integer degrees of freedom for the $X$ channel uses the novel concept of inteference-alignment \cite{MMKreport1, arxiv_dofx2,MMKreport2, Jafar_Shamai}. Interference alignment refers to the simple idea that signal vectors can be aligned in such a manner that they cast overlapping shadows at the receivers where they constitute interference while they continue to be distinct at the receivers where they are desired. The possibility of implicit interference alignment was first observed by Maddah-Ali, Motahari and Khandani in \cite{MMKreport1}. The first explicit interference alignment scheme was presented in \cite{arxiv_dofx2} where it was shown to be sufficient to achieve the full degrees of freedom for the MIMO $X$ channel. Interference alignment was subsequently used in \cite{MMKreport2,Jafar_Shamai} to show achievability of all points within the degrees of freedom region of the MIMO $X$ channel. Interference alignment was also independently discovered in the context of the compound broadcast channel in \cite{Weingarten_Shamai_Kramer}.

Since the distinction between the $X$ channel and the interference channel is quite significant, it is not immediately obvious whether the results found for the $X$ channel have any implications for the interference channel. For instance, the achievability schemes with inteference alignment proposed in \cite{MMKreport2} utilize the broadcast and multiple access channels inherent in the $X$ channel. However, the interference channel does not have broadcast and multiple access components as each transmitter has a message for only one unique receiver. Therefore, in this paper we answer the following question.

{\bf Question 3: } \emph{What are the degrees of freedom benefits from interference alignment on the $K$ user interference network?}

The degrees of freedom can be viewed as a capacity characterization that is accurate to within $o(\log(\rho))$ where $\rho$ represents the signal to noise ratio (SNR). In order to pursue increasingly accurate capacity characterizations, in this paper we explore the notion of $\mathcal{O}(1)$ capacity of a network. The $\mathcal{O}(1)$ capacity is an approximation accurate to within a bounded constant of the actual capacity region. The constant term can depend only on the channel gains and is independent of the transmit powers of the users. The $\mathcal{O}(1)$ capacity  is a more accurate description of the network capacity than the degrees of freedom of a network. Interestingly, for the point to point MIMO channel the $\mathcal{O}(1)$ capacity $\overline{C}(\rho)$ is directly related to the degrees of freedom $D$ as $\overline{C}(\rho)=D\log(1+\rho)$. This leads us to the third set of questions that we pursue in this paper.

{\bf Question 4: }\emph{Is the $\mathcal{O}(1)$ capacity $\overline{C}(\rho)$ of the  multiple access and  broadcast channels,  as well as the $2$ user interference and $X$ channels related to the degrees of freedom $D$ as $\overline{C}(\rho)=D\log(1+\rho)$? Does the same relationship hold for the $K$ user interference channel?}

Finally, we explore the benefits in terms of degrees of freedom, from the cognitive sharing of messages on the interference channel. Based on the cognitive radio model introduced in  \cite{Devroye_Mitran_Tarokh, Devroye_Mitran_Tarokh_Mag,Jovicic_Viswanath} cognitive message sharing refers to the form of cooperation where a message is made available non-causally to some transmitters and/or receivers besides the intended source and destination of the message. It was shown in \cite{Devroye_Sharif} that for the $2$ user interference channel with single antennas at each node, cognitive message sharing (from one transmitter to another) does not produce any gain in the degrees of freedom. The result was extended in \cite{Jafar_Shamai} to the $2$ user interference channel with \emph{multiple} antenna nodes and equal number of antennas at each node, to show that there is no gain in degrees of freedom whether a message is shared with the transmitter, receiver or both the transmitter and receiver of the other user. \cite{Jafar_Shamai} also establishes an interesting duality relationship where it is shown that from the degrees of freedom perspective  \emph{cognitive transmitters are equivalent to cognitive receivers}, i.e. sharing a message with another user's transmitter is equivalent to sharing a message with the other user's receiver. It is not clear if similar results will hold for the $3$ user interference channel, and it forms the last set of questions that we address in this paper.

{\bf Question 5: } \emph{For the $3$ user interference channel, what are the benefits of cognitive message sharing? Are cognitive transmitters equivalent to cognitive receivers in the manner shown for the $2$ user interference channel? }

\subsection{Overview of Results}
The answer to the first question is provided in Section \ref{sec:KK} and may be summarized in general terms as follows:

``\emph{Regardless of how many speakers and listeners are located within earshot of each other, each speaker can speak half the time and be heard without any interference by its intended listener}''

This result may seem impossible at first. For example, how can a total duration of $1$ hour be shared by 100 speakers such that each speaker speaks for $30$ minutes and is heard interference free by its intended listener when all systems are located within earshot of each other? And yet, this seemingly impossible result is made possible by the concept of interference alignment. A simple scheme is explained in Section \ref{sec:KK} where it is assumed the speakers and listeners can choose their locations. 

The answer to the second question is provided in Theorems \ref{theorem:dofach} and \ref{theorem:dofr}. We show that the $K$ user interference channel with single antennas at all nodes has (almost surely) a total of $K/2$ degrees of freedom per orthgonal time and frequency dimension when the channels are drawn randomly from a continuous distribution. The implications of this result for our understanding of the capacity of wireless networks are quite profound. It shows that we have grossly underestimated the capacity of wireless networks. For example, at high SNR the true capacity is higher by 50\%, 900\%, and 4900\% than anything previously shown to be achievable for networks with $3$, $20$, and $100$ interfering users, respectively. Interference is one of the principal challenges faced by wireless networks. However, we have shown that with perfect channel knowledge the frequency selective interference channel is not interference limited. In fact, after the first two users, additional users do not compete for degrees of freedom and each additional user is able to achieve $1/2$ degree of freedom without hurting the previously existing users. What makes this result even more remarkable is that linear scaling of degrees of freedom with users is achieved without cooperation in the form of message sharing that may allow MIMO behavior. Note that it has been shown previously for the $2$ user interference channel that unidirectional message sharing (e.g. from transmitter $1$ to transmitter $2$) does not allow higher degrees of freedom \cite{Devroye_Sharif,Jafar_Shamai} and even bi-directional message sharing (through full duplex noisy channels between the transmitters and full duplex noisy channels between the receivers) will not increase the degrees of freedom if the cost of message sharing is considered \cite{Nosratinia-Madsen, Madsen_CTW}. Therefore it is quite surprising that the $K$ user interference channel has $K/2$ degrees of freedom even without any message sharing. To summarize, Theorem \ref{theorem:dofach} shows that {\it only half the degrees of freedom are lost due to distributed processing at the transmitters and receivers on the interference channel}.


The answer to question 3 is provided by the achievability proof for Theorem \ref{theorem:dofach} where we find that, similar to the $2$ user $X$ channel, interference alignment suffices to achieve all the degrees of freedom on the $K$ user interference channel as well. Thus, interference alignment is as relevant for the $K$ user interference channel where it achieves the full $K/2$ degrees of freedom, as it is for the $2$ user $X$ channel where it achieves the full $4/3$ degrees of freedom. Interestingly, interference alignment does for wireless networks what MIMO technology has done for the point to point wireless channel. In both cases the capacity, originally limited to $\log(1+SNR)$, is shown to be capable of linearly increasing with the number of antennas. While MIMO technology requires nodes equipped with multiple antennas, interference alignment works with the distributed antennas naturally available in a network across the interfering transmitters and receivers.

\begin{figure}[!tbp]
\centerline{\input{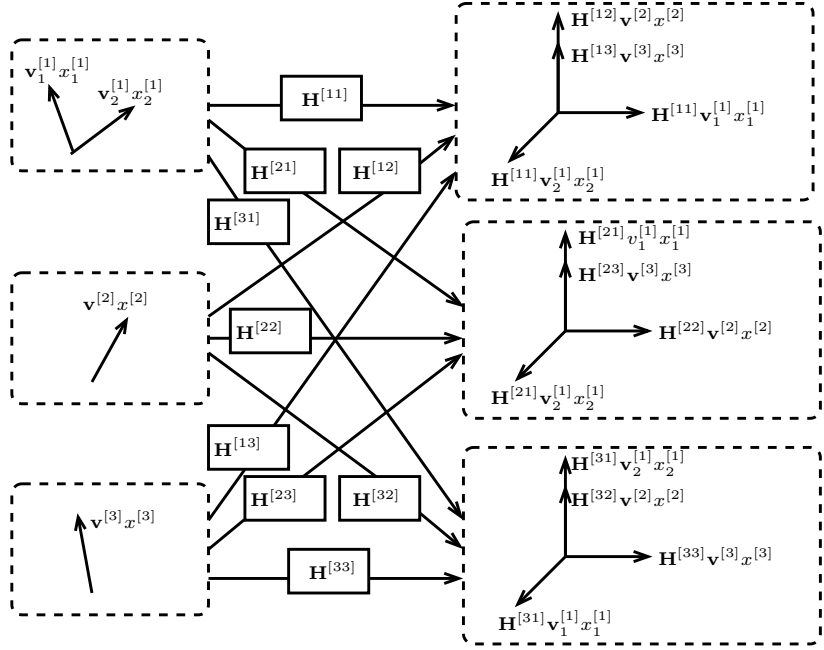}}
\caption{Interference alignment on the $3$ user interference channel to achieve $4/3$ degrees of freedom}
\label{figure:4by3}
\end{figure}

Figure \ref{figure:4by3} shows how interference alignment applies to the $3$ user interference channel. In this figure we illustrate how $4$ degrees of freedom are achieved over a 3 symbol extension of the channel with $3$ single antenna users, so that a total of $4/3$ degrees of freedom are achieved per channel use. The achievability proof for $3/2$ degrees of freedom is more involved and is provided in Section \ref{proof:dofach}. User $1$ achieves $2$ degrees of freedom by transmitting two independently coded streams along the beamforming vectors ${\bf v}^{[1]}_1, {\bf v}^{[1]}_2$ while users $2$ and $3$ achieve one degree of freedom by sending their independently encoded data streams along the beamforming vectors ${\bf v}^{[2]}, {\bf v}^{[3]}$, respectively. The beamforming vectors are chosen as follows.
\begin{itemize} 
\item At receiver $1$, the interference from transmitters $2$ and $3$ are perfectly aligned. 
\item At receiver $2$, the interference from transmitter $3$ aligns itself along one of the dimensions of the two-dimensional interference signal from transmitter $1$. 
\item Similarly, at receiver $3$, the interference from transmitter $2$ aligns itself along one of the dimensions of interference from transmitter $1$.
\end{itemize}

For the answer to question $4$, we show that for the multiple access, broadcast, and $2$ user interference and $X$ channels, the total degrees of freedom $D$ and the $\mathcal{O}(1)$ capacity $\overline{C}(\rho)$ are indeed related as $\overline{C}(\rho)=D\log(1+\rho)$. Thus, the two descriptions are equivalent. However, for the $3$ user interference channel with single antenna nodes it appears unlikely that such a relationship exists. The evidence in this paper raises the interesting possibility that the sum capacity of the $K$ user interference channel with single antenna nodes may have a different form than the multiple access, broadcast, and $2$ user interference and $X$ channels in that the difference between the true capacity $C(\rho)$ and the degrees of freedom approximation $D\log(1+\rho)$ may not be bounded. 

Finally, the answer to question $5$ is provided in Theorem \ref{theorem:cogint}. We show that sharing one message with all other transmitters and/or receivers does not increase the degrees of freedom for the $3$ user interference channel. Sharing two messages with all other transmitters and/or receivers on the other hand raises the degrees of freedom from $3/2$ to $2$. Interestingly, we find that the equivalance established between the cognitive transmitters and cognitive receivers on the $2$ user interference channel does not directly apply to the $3$ user interference channel. Intuitively, this may be understood as follows. A cognitive transmitter on the $3$ user interference channel can be more useful than a cognitive receiver. This is because a cognitive transmitter with no message of its own can still increase the degrees of freedom by canceling interference from its cognitively acquired message at other receivers. In other words a cognitive transmitter with no message of its own, still lends a transmit antenna to the transmitter whose message it shares. On the other hand, a cognitive receiver with no message of its own is useless.
\section{System Model}
Consider the $K$ user interference channel, comprised of $K$ transmitters and $K$ receivers. We assume coding may occur over multiple orthogonal frequency and time dimensions and the rates as well as the degrees of freedom are normalized by the number of orthogonal time and frequency dimensions. Each node is equipped with only one antenna (multiple antenna nodes are considered later in this paper). The channel output at the $k^{th}$ receiver over the $f^{th}$ frequency slot and the $t^{th}$ time slot is described as follows:
\begin{eqnarray*}
Y^{[k]}(f,t)&=&H^{[k1]}(f)X^{[1]}(f,t)+H^{[k2]}(f)X^{[2]}(f,t)+\cdots+H^{[kK]}(f)X^{[K]}(f,t)+Z^{[k]}(f,t)
\end{eqnarray*}
where, $k\in\{1,2,\cdots,K\}$ is the user index, $f\in\mathbb{N}$ is the frequency slot index, $t\in\mathbb{N}$ is the time slot index, $Y^{[k]}(f,t)$ is the output signal of the $k^{th}$ receiver, $X^{[k]}(f,t)$ is the input signal of the $k^{th}$ transmitter, $H^{[kj]}(f)$ is the channel fade coefficient from transmitter $j$ to receiver $k$ over the $f^{th}$ frequency slot and $Z^{[k]}(f,t)$ is the additive white Gaussian noise (AWGN) term at the $k^{th}$ receiver. The channel coefficients vary across frequency slots but are assumed constant in time. We assume all noise terms are i.i.d. (independent identically distributed) zero mean complex Gaussian with unit variance. We assume all channel coefficients $H^{[kj]}(f)$ are known a-priori to all transmitters and receivers. Note that since the channel coefficients do not vary in time, only causal channel knowledge is required. If we allow non-causal channel knowledge then the channel model above is equivalently represented as coding entirely in the time domain, i.e. over only one frequency slot. To avoid degenerate channel conditions (e.g. all channel coefficients are equal or channel coefficients are equal to either zero or infinity) we assume that the channel coefficient values are drawn i.i.d.  from a continuous distribution and the absolute value of all the channel coefficients is bounded between a non-zero minimum value and a finite maximum value. Since the channel values are assumed constant in time, the time index $t$ is sometimes suppressed for compact notation.

We assume that transmitters $1,2,\cdots,K$ have independent messages $W_1,W_2,\cdots,W_K$ intended for receivers $1,2,\cdots,K$, respectively. The total power across all transmitters is assumed to be equal to $\rho$ per orthogonal time and frequency dimension. We indicate the size of the message set by $|W_i(\rho)|$. For codewords spanning $f_0\times t_0$ channel uses (i.e. using $f_0$ frequency slots and $t_0$ time slots), the rates $R_i(\rho)=\frac{\log|W_i(\rho)|}{f_0t_0}$ are achievable if the probability of error for all messages can be simultaneously made arbitrarily small by choosing an appropriately large $f_0t_0$. 

The capacity region $\mathcal{C}(\rho)$ of the three user interference channel is the set of all \emph{achievable} rate tuples ${\bf R}(\rho)=(R_1(\rho), R_2(\rho), \cdots, R_K(\rho))$. 

\subsection{Degrees of Freedom}
Similar to the degrees of freedom region definition for the MIMO $X$ channel in \cite{Jafar_Shamai} we define the degrees of freedom region $\mathcal{D}$ for the $K$ user interference channel as follows:
\begin{align}
\mathcal{D}=\bigg\{&(d_1,d_2,\cdots, d_K)\in\mathbb{R}^K_+: \forall (w_1,w_2,\cdots,w_K)\in \mathbb{R}^K_+ \nonumber \\
&w_1d_1+w_2d_2+\cdots+w_Kd_K\leq \limsup_{\rho\rightarrow\infty}\left[\sup_{{\bf R}(\rho)\in\mathcal{C}(\rho)}[w_1R_1(\rho)+w_2R_2(\rho)+\cdots+w_KR_K(\rho)]\frac{1}{\log(\rho)}\right]\bigg\}
\end{align}

\section{Interference Alignment through Channel Design}\label{sec:KK}
With the exception of this section, throughout this paper we assume that the channel coefficients are determined by nature, i.e. we do not control the channel values, and we only control the coding scheme, i.e. the transmitted symbols. However, in this section we take a different perspective to gain additional insights into the problem. We wish to know what is the best we can do if we are allowed to pick all the channel coefficient values subject to the only constraint that the coefficient values are finite, non-zero constants. It is important that we can only pick non-zero channel coefficient values because the $K/2$ outerbound applies if and only if all channel coefficients have non-zero values. For example, if we are allowed to set some channel coefficients to zero the problem becomes trivial because by setting all interfering links to zero we can easily achieve $K$ degrees of freedom over the $K$ non-interfering channels. 

\subsection{Interference Alignment by Choice of Channel Coefficients}
As we show next, we can achieve $K/2$ degrees of freedom for the $K$ user interference channel with non-zero channel coefficients if we are allowed to pick the values of the channel coefficients. The proof is quite simple. We consider a two symbol extension of the channel, i.e. coding over two frequency slots, where the channel is defined by $2\times2$ diagonal channel matrices that we choose as follows
\begin{eqnarray}
{\bf H}^{[ij]}&=&\left[\begin{array}{cc}1&0\\0&-1\end{array}\right] ~~\mbox{if} ~~i\neq j\\
{\bf H}^{[ij]}&=&\left[\begin{array}{cc}1&0\\0&1\end{array}\right] ~~\mbox{if} ~~i=j
\end{eqnarray}
Each user transmits his coded symbols along the beamforming vector
\begin{eqnarray}
{\bf v}^{[i]}=\left[\begin{array}{c}1\\1\end{array}\right]
\end{eqnarray}
This ensures that all the interference terms at each receiver appear along the direction vector $[1~~~ -1]^T$ while the desired signal at each receiver appears along the direction $[1 ~~~ 1]^T$. Thus the desired signal and interference are orthogonal so that each user is able to achieve one degree of freedom for his message. Since $K$ degrees of freedom are achieved over the $2$ symbol extension of the channel the degrees of freedom equal $K/2$. Thus, it is interesting to note that the $K/2$ outerbound is tight for some interference channels with non-zero channel coefficients. Since joint processing at all transmitters and at all receivers would result in $K$ degrees of freedom on the $K$ user interference channel, we observe that if we are allowed to pick the channel coefficients then \emph{the maximum penalty for distributed signal processing is the loss of half the degrees of freedom}. It remains to be shown if this bound is tight when channel coefficients are chosen by nature, i.e. modeled as random variables drawn from a continuous distribution. As we show in the next section, the outerbound of $K/2$ is almost surely tight for the $K$ user interference channel.

\subsection{Interference alignment through choice of propagation delays - Can everyone speak half the time with no interference?}
We end this section with another interesting example of interference alignment.
Consider the $K$ user interference channel where there is a propagation delay from each transmitter to each receiver. Let $T_{ij}$ represent the signal propagation delay from transmitter $i$ to receiver $j$. Suppose the locations of the transmitters and receivers can be configured such that the delay $T_{ii}$ from each transmitter  to its intended receiver is an even multiple of a basic symbol duration $T_s$, while the signal propagation delays $T_{ij}, (i\neq j)$ from each transmitter to all unintended receivers are odd multiples of the symbol duration. The communication strategy is the following. All transmissions occur simultaneously at even symbol durations. Note that with this policy, each receiver sees its own transmitter's signal interference-free over even time periods, while it sees all interfering signals simultaneously over odd time periods. Thus \emph{each user is able to achieve $1/2$ degrees of freedom and the total degrees of freedom achieved is equal to $K/2$}. 

\section{ Degrees of Freedom for the $K$ User Interference Channel - Interference Alignment through Precoding }
Henceforth, we assume that the channel coefficients are not controlled by the nodes but rather selected by nature. Thus, the channels do not automatically align the interference and any interference alignment can only be accomplished through code design. The following theorem presents the main result of this section.
\begin{theorem}\label{theorem:dofach}
The number of degrees of freedom for the $K$ user interference channel with single antennas at all nodes is $K/2$.
\begin{eqnarray}
\max_{{\bf d}\in\mathcal{D}} d_1+d_2+\cdots+d_K&=& K/2
\end{eqnarray}
\end{theorem}

The converse argument for the theorem follows directly from the outerbound for the $K$ user interference channel presented in \cite{Nosratinia-Madsen}. The achievability proof is presented next. Since the proof is rather involved, we present first the constructive proof for $K=3$. The proof for general $K\geq 3$ is then provided in Appendix \ref{app:K}.
\subsection{Achievability Proof for Theorem \ref{theorem:dofach} with $K=3$}\label{proof:dofach}
We show that $(d_1,d_2,d_3) = (\frac{n+1}{2n+1}, \frac{n}{2n+1}, \frac{n}{2n+1})$ lies in the degrees of freedom region $\forall n \in \mathbb{N}$. Since the degrees of freedom region is closed, this automatically implies that 
$$\max_{(d_1,d_2,d_3) \in \mathcal{D}} d_1+d_2+d_3 \geq \sup_{n} \frac{3n+1}{2n+1} = \frac{3}{2}$$ 
This result, in conjunction with the converse argument proves the theorem.

To show that $ (\frac{n+1}{2n+1}, \frac{n}{2n+1}, \frac{n}{2n+1})$ lies in $\mathcal{D}$, we construct an interference alignment scheme using only $2n+1$ frequency slots. We collectively denote the $2n+1$ symbols transmitted over the first $2n+1$ frequency slots at each time instant as a supersymbol. We call this the $(2n+1)$ symbol extension of the channel. With the extended channel, the signal vector at the $k^{th}$ user's receiver can be expressed as 
$$ \barxY^{[k]}= \barxH^{[k1]} \barxX^{[1]} + \barxH^{[k2]}\barxX^{[2]} + \barxH^{[k3]}\barxX^{[3]} + \barxZ^{[k]}, ~~k=1,2,3.$$
where $\barxX^{[k]}$ is a $(2n+1) \times 1$ column vector representing the $2n+1$ symbol extension of the transmitted symbol $X^{[k]}$, i.e 
$$\barxX^{[k]}(t) \define \left[ \begin{array}{c} X^{[k]}(1,t) \\X^{[k]}(2,t)\\ \vdots \\ X^{[k]}(2n+1,t) \end{array}\right]$$ Similarly $\barxY^{[k]}$ and $\barxZ^{[k]}$ represent $2n+1$ symbol extensions of the $Y^{[k]}$ and $Z^{[k]}$ respectively. 
$\barxH^{[kj]}$ is a diagonal $(2n+1)\times(2n+1)$ matrix representing the $2n+1$ symbol extension of the channel i.e  
$$ \barxH^{[kj]} \define \left[ \begin{array}{cccc}  H^{[kj]}(1) & 0 & \ldots & 0\\
	0 & H^{[kj]}(2) & \ldots & 0\\
	\vdots & \cdots & \ddots & \vdots\\ 
	0 & 0& \cdots  & H^{[kj]}(2n+1) \end{array}\right] $$
Recall that we assume that the channel coefficient values for each frequency slot are chosen independently from a continuous distribution. Thus, all the diagonal channel matrices $\barxH^{[kj]}$ are comprised of all distinct diagonal elements with probability $1$.

We show that $(d_1,d_2,d_3)=(n+1,n,n)$ is achievable on this extended channel implying that $(\frac{n+1}{2n+1}, \frac{n}{2n+1}, \frac{n}{2n+1})$ lies in the degrees of freedom region of the original channel. 

In the extended channel, message $W_1$ is encoded at transmitter $1$ into $n+1$ independent streams $x^{[1]}_m(t), m=1,2,\ldots,(n+1)$ sent along vectors $\mathbf{v}^{[1]}_m$ so that $\barxX^{[1]}(t)$ is
$$\barxX^{[1]}(t) = \displaystyle\sum_{m=1}^{n+1} x^{[1]}_m(t) \mathbf{v}_m^{[1]} = \barxV^{[1]} \xX^{[1]}(t)$$
where $\xX^{[1]}(t)$ is a $(n+1) \times 1$ column vector and $\barxV^{[1]}$ is a $(2n+1) \times (n+1)$ dimensional matrix.
Similarly $W_2$ and $W_3$ are each encoded into $n$ independent streams by transmitters $2$ and $3$ as $\xX^{[2]}(t)$ and $\xX^{[3]}(t)$ respectively.
$$\barxX^{[2]}(t) = \displaystyle\sum_{m=1}^{n}{x}^{[2]}_m(t) \mathbf{v}_m^{[2]} = \barxV^{[2]}\xX^{[2]}(t)$$
$$\barxX^{[3]}(t) = \displaystyle\sum_{m=1}^{n}{x}^{[3]}_m(t) \mathbf{v}_m^{[3]} = \barxV^{[3]} \xX^{[3]}(t)$$
The received signal at the $i^{th}$ receiver can then be written as
$$ \barxY^{[i]}(t) = \barxH^{[i1]} \barxV^{[1]} \xX^{[1]}(t) + \barxH^{[i2]} \barxV^{[2]} \xX^{[2]}(t) + \barxH^{[i3]}\barxV^{[3]} \xX^{[3]}(t) + \barxZ^{[i]}(t)$$

In this achievable scheme, receiver $i$ eliminates interference by zero-forcing all $\barxV^{[j]},j \neq i$ to decode $W_i$. At receiver 1, $n+1$ desired streams are decoded after zero-forcing the interference to achieve $n+1$ degrees of freedom. To obtain $n+1$ interference free dimensions from a $2n+1$ dimensional received signal vector ${\barxY}^{[1]}(t)$, the dimension of the interference should be not more than $n$. This can be ensured by perfectly aligning the interference from transmitters $2$ and $3$ as follows.
\begin{equation} \label{achint:rx1} \barxH^{[12]} \barxV^{[2]} = \barxH^{[13]} \barxV^{[3]} \end{equation}
At the same time,  receiver $2$ zero-forces the interference from $\barxX^{[1]}$ and $\barxX^{[3]}$. To extract $n$ interference-free dimensions from a $2n+1$ dimensional vector, the dimension of the interference has to be not more than $n+1$.
i.e. $$ \mbox{rank}\left(\left[\barxH^{[21]} \barxV^{[1]} ~~~~\barxH^{[23]} \barxV^{[3]}\right]\right) \leq n+1$$
This can be achieved by choosing $\barxV^{[3]}$ and $\barxV^{[1]}$ so that
\begin{equation}\label{achint:rx2} \barxH^{[23]} \barxV^{[3]} \prec \barxH^{[21]} \barxV^{[1]} \end{equation}
where $\mathbf{P} \prec \mathbf{Q}$, means that the set of column vectors of matrix $\mathbf{P}$ is a subset of the set of column vectors of matrix $\mathbf{Q}$.
Similarly, to decode $W_3$ at receiver $3$, we wish to choose $\barxV^{[2]}$ and $\barxV^{[1]}$ so that
\begin{equation} \label{achint:rx3} \barxH^{[32]} \barxV^{[2]} \prec \barxH^{[31]} \barxV^{[1]} \end{equation}
Thus, we wish to pick vectors $\barxV^{[1]}$, $\barxV^{[2]}$ and $\barxV^{[3]}$ so that equations (\ref{achint:rx1}), (\ref{achint:rx2}), (\ref{achint:rx3}) are satisfied.
Note that the channel matrices $\barxH^{[ij]}$ have a full rank of $2n+1$ almost surely. Since multiplying by a full rank matrix (or its inverse) does not affect the conditions represented by equations (\ref{achint:rx1}), (\ref{achint:rx2}) and (\ref{achint:rx3}), they can be equivalently expressed as
\begin{eqnarray}
 \label{achint:rxfirst}
 \mathbf{B} &=& \mathbf{T}\mathbf{C}\\
 \mathbf{B} &\prec& \mathbf{A} \\
 \mathbf{C} &\prec& \mathbf{A} 
 \label{achint:rxlast}
\end{eqnarray}
where 
\begin{eqnarray}
 \label{achint:ABCTfirst}
 \mathbf{A}&=& \barxV^{[1]}\\
 \mathbf{B}&=& (\barxH^{[21]})^{-1} \barxH^{[23]} \barxV^{[3]}\\
 \mathbf{C}&=& (\barxH^{[31]})^{-1} \barxH^{[32]} \barxV^{[2]}\\
 \mathbf{T}&=& \barxH^{[12]} (\barxH^{[21]})^{-1}\barxH^{[23]} (\barxH^{[32]})^{-1}\barxH^{[31]} (\barxH^{[13]})^{-1}
 \label{achint:ABCTlast}
\end{eqnarray}
Note that $\mathbf{A}$ is a $(2n+1) \times (n+1)$ matrix. $\mathbf{B}$ and $\mathbf{C}$ are $(2n+1) \times n$ matrices. Since all channel matrices are invertible, we can choose $\mathbf{A}$, $\mathbf{B}$ and $\mathbf{C}$ so that they satisfy equations (\ref{achint:rxfirst})-(\ref{achint:rxlast}) and then use equations (\ref{achint:ABCTfirst})-(\ref{achint:ABCTlast})  to find $\barxV^{[1]}$,$\barxV^{[2]}$ and $\barxV^{[3]}$. $\mathbf{A}$, $\mathbf{B}$, $\mathbf{C}$ are picked as follows. Let $\mathbf{w}$ be the $(2n+1)\times 1$ column vector 
$$\mathbf{w} = \left[ \begin{array}{c} 1 \\ 1 \\ \vdots \\ 1\end{array} \right]$$
We now choose $\mathbf{A}$, $\mathbf{B}$ and $\mathbf{C}$ as:
\begin{eqnarray*}
\mathbf{A}&=&[\xw\hspace{6pt}\mathbf{T}\xw \hspace{6pt}\mathbf{T}^2\xw \hspace{6pt} \ldots \hspace{6pt} \mathbf{T}^{n}\xw]\\
\mathbf{B}&=&[\mathbf{T}\xw\hspace{6pt}\mathbf{T}^2\xw\hspace{6pt}  \ldots \hspace{6pt} \mathbf{T}^{n}\xw]\\
\mathbf{C}&=&[\xw\hspace{6pt}\mathbf{T}\xw\hspace{6pt}  \ldots\hspace{6pt} \mathbf{T}^{n-1}\xw]
\end{eqnarray*}
It can be easily verified that $\mathbf{A}$, $\mathbf{B}$ and $\mathbf{C}$ satisfy the three equations (\ref{achint:rxfirst})-(\ref{achint:rxlast}). Therefore, $\barxV^{[1]}$, $\barxV^{[2]}$ and $\barxV^{[3]}$ satisfy the interference alignment equations in (\ref{achint:rx1}), (\ref{achint:rx2}) and (\ref{achint:rx3}).

Now, consider the received signal vectors at Receiver $1$. The desired signal arrives along the $n+1$ vectors $\barxH^{[11]} \barxV^{[1]}$ while the interference arrives along the $n$ vectors $\barxH^{[12]} \barxV^{[2]}$ and  the $n$ vectors $\barxH^{[13]} \barxV^{[3]}$. As enforced by equation (\ref{achint:rx1}) the interference vectors are perfectly aligned. Therefore, in order to prove that there are $n+1$ interference free dimensions it suffices to show that the columns of the square, $(2n+1)\times (2n+1)$ dimensional matrix
\begin{eqnarray}
\left[\barxH^{[11]} \barxV^{[1]}~~~~\barxH^{[12]} \barxV^{[2]}\right]
\end{eqnarray}
are linearly independent almost surely. Multiplying by the full rank matrix $(\barxH^{[11]})^{-1}$ and substituting the values of $\barxV^{[1]}, \barxV^{[2]}$, equivalently we need to show that almost surely
\begin{eqnarray}
{\bf S}\define\left[\xw\hspace{6pt}\mathbf{T}\xw \hspace{6pt}\mathbf{T}^2\xw \hspace{6pt} \ldots \hspace{6pt} \mathbf{T}^{n}\xw 
\hspace{6pt} {\bf D}\xw\hspace{6pt}{\bf D}\mathbf{T}\xw \hspace{6pt}{\bf D}\mathbf{T}^2\xw \hspace{6pt} \ldots \hspace{6pt} {\bf D}\mathbf{T}^{n-1}\xw
\right]
\end{eqnarray}
has linearly independent column vectors where ${\bf D}=(\barxH^{[11]})^{-1}\barxH^{[12]}$ is a diagonal matrix. In other words, we need to show $\det({\bf S})\neq 0$ with probability 1. The proof is obtained by contradiction. If possible, let $\mathbf{S}$ be singular with non-zero probability. i.e, $\Pr(|\mathbf{S}|=0) > 0$. Further, let the diagonal entries of ${\bf T}$ be $\lambda_1, \lambda_2, \ldots \lambda_{2n+1}$ and the diagonal entries of ${\bf D}$ be $\kappa_{1}, \kappa_{2} \ldots \kappa_{2n+1}$. Then the following equation is true with non-zero probability.
\begin{eqnarray*}
|\mathbf{S}| = 
\left|
\begin{array}{ccccccccccccc}
1&\lambda_1 & \lambda_1^2 & \ldots & \lambda_1^{n} & \kappa_1 & \kappa_1 \lambda_1 & \ldots &\kappa_1 \lambda_1^{n-1}\\
1&\lambda_2 & \lambda_2^2 & \ldots & \lambda_2^{n} & \kappa_2 & \kappa_2 \lambda_2 & \ldots &\kappa_2 \lambda_2^{n-1}\\
\vdots  &  \vdots & \vdots & \ddots & \vdots & \vdots & \vdots &\ddots & \vdots \\ 
1&\lambda_{2n+1} & \lambda_{2n+1}^2 &  \ldots & \lambda_{2n+1}^{n} & \kappa_{2n+1} & \kappa_{2n+1} \lambda_{2n+1} &\ldots& \kappa_{2n+1} \lambda_{2n+1}^{n-1}\end{array} \right|&=& 0
\end{eqnarray*}
Let $C_{ij}$ indicate the cofactor of the $i$th row and $j$th column of $|\mathbf{S}|$.
Expanding the determinant along the first row, we get
$$ |\mathbf{S}|=0 \Rightarrow C_{11} + \lambda_1C_{12} + \ldots \lambda_1^{n} C_{1(n+1)} + \kappa_1 \left[C_{1(n+2)} + \lambda_1C_{1(n+3)} +\ldots + \lambda_1^{n-1} C_{1(2n+1)}\right] = 0  $$
None of `co-factor' terms $C_{1j}$ in the above expansion depend $\lambda_1$ and $\kappa_1$. If all values other than $\kappa_1$ are given, then the above is a linear equation in $\kappa_1$. Now, $|\mathbf{S}|=0$ implies one of the following two events
\begin{enumerate}
\item $\kappa_1$ is a root of the linear equation.
\item All the coefficients forming the linear equation in $\kappa_1$ are equal to $0$, so that the singularity condition is trivially satisfied for all values of $\kappa_1$.
\end{enumerate}
Since $\kappa_1$ is a random variable drawn from a continuous distribution, the probability of $\kappa_1$ taking a value which is equal to the root of this linear equation is zero. Therefore, the second event happens with probability greater than $0$ and we can write,
\begin{eqnarray*}\Pr\left(|\mathbf{S}| = 0\right) > 0 &\Rightarrow&  \Pr(C_{1(n+2)} + \lambda_1C_{1(n+3)} +\ldots + \lambda_1^{n-1} C_{1(2n+1)}= 0) > 0 \end{eqnarray*}
Consider the equation
\begin{eqnarray*}C_{1(n+2)} + \lambda_1C_{1(n+3)} +\ldots + \lambda_1^{n-1} C_{1(2n+1)} = 0 
\end{eqnarray*}
Since the terms $C_{1j}$ do not depend on $\lambda_1$, the above equation is a polynomial of degree $n$ in $\lambda_1$. Again, as before, there are two possibilities. The first possibility is that $\lambda_1$ takes a value equal to one of the $n$ roots of the above equation. Since $\lambda_1$ is drawn from a continuous distribution, the probability of this event happening is zero. 
The second possibility is that all the coefficients of the above polynomial are zero with non-zero probability and we can write
\begin{eqnarray*}\Pr(C_{1(n+2)} + \ldots + \kappa_1 \lambda_1^{n} C_{1(2n+1)} = 0) > 0 \Rightarrow \Pr(C_{1(2n+1)} = 0) > 0\end{eqnarray*}
We have now shown that if the determinant of the $(2n+1) \times (2n+1)$ matrix ${\bf S}$ is equal to $0$ with non-zero probability, then the determinant of following $2n\times2n$ matrix (obtained by stripping off the first row and last column of ${\bf S})$ is equal to $0$ with non-zero probability.
$$\det
\left[
\begin{array}{ccccccccccccc}
1&\lambda_2 & \lambda_2^2 & \ldots & \lambda_2^{n} & \kappa_2 & \kappa_2 \lambda_2 & \ldots &\kappa_2 \lambda_2^{n-2}\\
\vdots  &  \vdots & \vdots & \ddots & \vdots & \vdots & \vdots &\ddots & \vdots \\ 
1&\lambda_{2n+1} & \lambda_{2n+1}^2 &  \ldots & \lambda_{2n+1}^{n} & \kappa_{2n+1} & \kappa_{2n+1} \lambda_{2n+1} &\ldots& \kappa_{2n+1} \lambda_{2n+1}^{n-2}\end{array} \right] = 0 
$$
with probability greater than $0$.
Repeating the above argument and eliminating the first row and last column at each stage we get
\begin{eqnarray*}
\det\left[
\begin{array}{ccccccccccccc}
1&\lambda_{n+1} & \lambda_{n+1}^2 & \ldots & \lambda_{n+1}^{n} \\
\vdots  &  \vdots & \vdots & \ddots & \vdots & \\
1&\lambda_{2n+1} & \lambda_{2n+1}^2 &  \ldots & \lambda_{2n+1}^{n} \end{array} \right]=0
\end{eqnarray*}
with probability greater than $0$. But this is a Vandermonde matrix and its determinant  
$$\prod_{n+1\leq i < j\leq2n+1} (\lambda_i - \lambda_j)$$
is equal to $0$ only if $\lambda_i = \lambda_j$ for some $i \neq j$. Since $\lambda_i$ are drawn independently from a continuous distribution, they are all distinct almost surely. This implies that $\Pr(|\mathbf{S}| = 0 )  = 0 $.

Thus, the $n+1$ vectors carrying the desired signal at receiver $1$ are linearly independent of the $n$ interference vectors which allows the receiver to zero force interference and obtain $n+1$ interference free dimensions, and therefore $n+1$ degrees of freedom for its message.

At receiver $2$ the desired signal arrives along the $n$ vectors $\barxH^{[22]} \barxV^{[2]}$ while the interference arrives along the $n+1$ vectors $\barxH^{[21]} \barxV^{[1]}$ and  the $n$ vectors $\barxH^{[23]} \barxV^{[3]}$. As enforced by equation (\ref{achint:rx2}) the interference vectors $\barxH^{[23]} \barxV^{[3]}$ are perfectly aligned within the interference vectors $\barxH^{[21]} \barxV^{[1]}$. Therefore, in order to prove that there are $n$ interference free dimensions at receiver $2$ it suffices to show that the columns of the square, $(2n+1)\times (2n+1)$ dimensional matrix
\begin{eqnarray}
\left[\barxH^{[22]} \barxV^{[2]}~~~~\barxH^{[21]} \barxV^{[1]}\right]
\end{eqnarray}
are linearly independent almost surely. This proof is quite similar to the proof presented above for receiver $1$ and is therefore omitted to avoid repetition. Using the same arguments we can show that both receivers $2$ and $3$ are able to zero force the $n+1$ interference vectors and obtain $n$ interference free dimensions for their respective desired signals so that they each achieve $n$ degrees of freedom.

Thus we established the achievability of $d_1+d_2+d_3 = \frac{3n+1}{2n+1}$ for any $n$. This scheme, along with the converse automatically imply that
$$\sup_{(d_1,d_2,d_3) \in \mathcal{D}}{d_1+d_2+d_3} = \frac{3}{2}$$

\subsection{The Degrees of Freedom Region for the $3$ User Interference Channel}
\begin{theorem}\label{theorem:dofr}
The degrees of freedom region of the $3$ user interference channel is characterized as follows:
\begin{eqnarray}
\mathcal{D}=\left\{(d_1,d_2,d_3):\right.&& \nonumber\\
\label{eqn:thmdofrfirst}
d_1+d_2&\leq& 1\nonumber\\
d_2+d_3&\leq&1\nonumber\\
d_1+d_3&\leq& \left.1\right\}
\label{eqn:thmdofrlast}
\end{eqnarray}
\end{theorem}
\begin{figure}
\centerline{\input{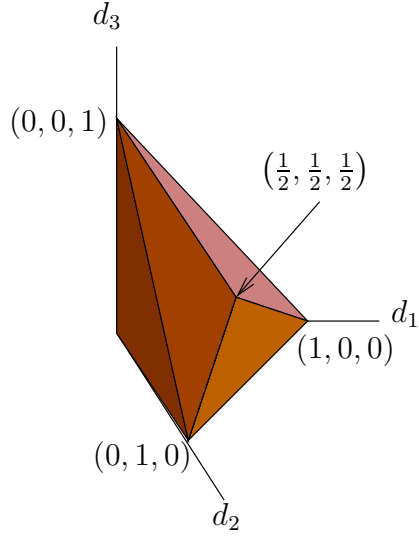}}

\caption{Degrees of Freedom Region for the $3$ user interference channel}\label{fig:dofrint3}
\end{figure}

\begin{proof}
The converse argument is identical to the converse argument for Theorem \ref{theorem:dofach} and is therefore omitted.  We show achievability as follows. Let $\mathcal{D}^{'}$ be the degrees of freedom region of the $3$ user interference channel. We need to prove that $\mathcal{D}^{'} = \mathcal{D}$. We show that $\mathcal{D} \subset \mathcal{D}^{'}$ which along with the converse proves the stated result.

The points $K=(0,0,1)$, $L=(0,1,0)$, $J=(1,0,0)$ can be verified to lie in $\mathcal{D}^{'}$ through trivial achievable schemes. Also, Theorem \ref{theorem:dofach} implies that $N=(\frac{1}{2},\frac{1}{2},\frac{1}{2})$ lies in $\mathcal{D}^{'}$ (Note that this is the only point which achieves a total of $\frac{3}{2}$ degrees of freedom and satisfies the inequalities in (\ref{eqn:thmdofrlast}).
Consider any point $(d_1,d_2,d_3) \in \mathcal{D}$ as defined by the statement of the theorem. The point $(d_1,d_2,d_3)$ can then be shown to lie in a convex region whose corner points are $(0,0,0)$, J, K, L and N. i.e  $(d_1,d_2,d_3)$ can be expressed as a convex combination of the end points (see Fig. \ref{fig:dofrint3}).

$$(d_1,d_2,d_3) = \alpha_1(1,0,0)+\alpha_2(0,1,0)+\alpha_3(0,0,1) + \alpha_4(\frac{1}{2},\frac{1}{2},\frac{1}{2}) + \alpha_5(0,0,0)$$
where the constants $\alpha_i$ are defined as follows.
\begin{eqnarray*}
\begin{array}{|c|c|c|c|c|c|}\hline
 & \alpha_1 & \alpha_2 & \alpha_3 & \alpha_4& \alpha_5\\
\hline 
d_1+d_2+d_3\leq 1 & d_1 & d_2 & d_3 & 0 & 1-d_1-d_2-d_3\\ \hline
d_1+d_2+d_3 > 1 & \frac{d_1-d_2-d_3+1}{2} & \frac{d_2-d_1-d_3+1}{2} & \frac{d_3-d_1-d_2+1}{2} & d_1+d_2+d_3-1 & 0\\ \hline
\end{array}
\end{eqnarray*}
It is easily verified that the values of $\alpha_i$ are non-negative for all $(d_1,d_2,d_3)\in\mathcal{D}$ and that they add up to one. Thus, all points in $\mathcal{D}$ are convex combinations of achievable points $J, K, L, N$ and $(0,0,0)$. Since convex combinations are achievable by time sharing between the end points, this implies that  $\mathcal{D} \subset \mathcal{D}^{'}$. Together with the converse, we have $ \mathcal{D} = \mathcal{D}^{'}$ and the proof is complete.
\end{proof}

Note that the proof presented above uses coding over multiple frequency slots where the channel coefficients take distinct values. We now examine the possible ramifications of this assumption both from a theoretical as well as a practical perspective. 

From a theoretical perspective the assumption of frequency selective channels is intriguing because it is not clear if $K/2$ degrees of freedom will be achieved with constant channels over only one frequency slot. Therefore the validity of the conjecture in \cite{Nosratinia-Madsen} that the interference channel with constant channel coefficients has only $1$ degree of freedom for any number of users still remains undetermined. The issue is analogous to the $2$ user  $X$ channel with a single antenna at all nodes. It is shown in \cite{Jafar_Shamai} that the time/frequency varying MIMO $X$ channel has $4/3$ degrees of freedom per time/frequency dimension. However, it is not known whether the $X$ channel with single antenna nodes and constant channel coefficients can achieve more than $1$ degree of freedom.

From a practical perspective, we present several observations. 
\begin{enumerate}
\item The assumption that the channel coefficients vary over frequency is not restrictive as it holds true in practice for almost all wireless channels. Moreover, note that it is not necessary that the channel coefficients are independent across frequency slots. It suffices that they are chosen according to a continuous joint distribution.
\item We have shown that by coding over $2n+1$ frequency slots, we can achieve $\frac{3n+1}{2n+1}$ degrees of freedom on the $3$ user interference channel. The fact that only a finite number of frequency slots suffice to achieve a certain number of degrees of freedom may be significant in practice. On the other hand if non-causal channel knowledge is not an issue and the channel is time varying then only one frequency slot suffices for this achievability proof.

\item Recall that for the $2$ user $X$ channel, time and frequency variations are not needed when more than $1$ antenna is present at each node. Similarly, we will show in Section \ref{section:MIMOint} that with $M>1$ antennas at each node the $3$ user interference channel with constant channel matrices has $3M/2$ degrees of freedom. 
\end{enumerate}
Before considering the MIMO case with constant channel matrices we visit the issue of $\mathcal{O}(1)$ capacity. 
\section{The $\mathcal{O}(1)$ Capacity of Wireless Networks}
Consider a multiuser wireless channel with transmit power $\rho$, noise power normalized to unity, and sum capacity $C(\rho)$. The degrees of freedom $d$ provide a capacity approximation that is accurate within $o(\log(\rho))$, i.e.,
\begin{eqnarray}
C(\rho)=d\log(\rho)+o(\log(\rho))
\end{eqnarray}
where the little "o" notation is defined as follows:
\begin{eqnarray}
f(x)=o(g(x))\Leftrightarrow \lim_{x\rightarrow \infty}\frac{f(x)}{g(x)}=0.
\end{eqnarray}




Similarly, one can define a capacity characterization $\overline{C}_1(\rho)$, that is accurate to within an $\mathcal{O}(1)$ term,
\begin{eqnarray}
\limsup_{\rho\rightarrow\infty}\left|C(\rho)-\overline{C}_1(\rho)\right|<\infty.
\end{eqnarray}
so that we can write
\begin{eqnarray}
C(\rho)=\overline{C}_1(\rho)+\mathcal{O}(1)
\end{eqnarray}


While the $\mathcal{O}(1)$ notation implies an asymptotic approximation as $\rho\rightarrow\infty$, it is easy to see that for all communication networks, if the $\mathcal{O}(1)$ capacity characterization is known, then one can find a capacity characterization that is within a constant of the capacity for \emph{all} $\rho$.  This is because the capacity $C(\rho)$ is a non-negative, monotonically increasing function of the transmit power $\rho$.  This is seen as follows. Let $\overline{C_1}(\rho)$ be an $\mathcal{O}(1)$ capacity characterization. Mathematically, $\exists \rho_o, C_o<\infty$, such that 
\begin{eqnarray}
\sup_{\rho\geq\rho_o}\left|C(\rho)-\overline{C}_1(\rho)\right|<C_o.
\end{eqnarray}
Then we can construct a capacity characterization $\overline{C}(\rho)$ that is accurate to within a constant for all $\rho$ as follows:
$$ \overline{C}(\rho) \define \left\{ \begin{array}{ll} \overline{C_1}(\rho_o) & \forall \rho \leq \rho_o \\ \overline{C_1}(\rho) &\forall \rho > \rho_o\end{array} \right\}$$
such that the absolute value of the difference between the capacity $C(\rho)$ and $\overline{C}(\rho)$ is bounded above by $\max\{C_o, \overline{C}_1(\rho)\}$.

Clearly, the $\mathcal{O}(1)$ capacity provides in general a more accurate capacity characterization than the degrees of freedom definition. However, it turns out that in most cases the two are directly related. For example, it is well known that for the full rank MIMO channel  with $M$ input antennas and $N$ output antennas, transmit power $\rho$ and i.i.d. zero mean unit variance additive white Gaussian noise (AWGN) at each receiver, the capacity $C(\rho)$  may be expressed as:
\begin{eqnarray}
C(\rho)=\min(M,N)\log(1+\rho)+\mathcal{O}(1)=d\log(1+\rho)+\mathcal{O}(1).
\end{eqnarray}

As formalized by the following theorem, a similar relationship between the degrees of freedom and the $\mathcal{O}(1)$ capacity characterization also holds for most multiuser communication channels. 
\begin{theorem}
For the MIMO multiple access channel, the MIMO broadcast channel, the two user MIMO interference channel and the 2 user MIMO $X$ channel, an $\mathcal{O}(1)$ characterization of the sum capacity can be obtained in terms of the total number of degrees of freedom as follows:
\begin{eqnarray}
C(\rho)=d\log(1+\rho)+\mathcal{O}(1).
\end{eqnarray}
\end{theorem}
\begin{proof}
Since the proof is quite simple, we only present a brief outline as follows. For the MIMO MAC and BC, the outerbound on sum capacity obtained from full cooperation among the distributed nodes is $d\log(1+\rho)+\mathcal{O}(1)$. The innerbound obtained from zero forcing is also $d\log(1+\rho)+\mathcal{O}(1)$ so that we can write $C(\rho)=d\log(1+\rho)+\mathcal{O}(1)$. For the two user MIMO interference channel and the 2 user MIMO $X$ channel the outerbound is obtained following an extension of Carlieal's outerbound which results in a MIMO MAC channel. The innerbound is obtained from zero forcing. Since both of these bounds are within $\mathcal{O}(1)$ of $d\log(1+\rho)$ we can similarly write $C(\rho)=d\log(1+\rho)+\mathcal{O}(1)$. 
\end{proof}

Finally, consider the  $K$ user interference channel with single antennas at each node. In this case we have only shown:
\begin{eqnarray}
(K/2-\epsilon)\log(1+\rho)+\mathcal{O}(1)\leq C(\rho) \leq (K/2)\log(1+\rho)+\mathcal{O}(1), ~\forall \epsilon>0.
\end{eqnarray}
Consider a hypothetical capacity function $C(\rho)=K/2\log(1+\rho)-c\sqrt{\log(1+\rho^2)}$. Such a capacity function would also satisfy the inner and outerbounds provided above for the $K$ user interference channel and has $D=K/2$ degrees of freedom. However, this hypothetical capacity function does not have a $\mathcal{O}(1)$ capacity characterization equal to $\overline{C}(\rho)=K/2\log(1+\rho)$ as the difference between $C(\rho)$ and $\overline{C}(\rho)$ is unbounded. To claim that  the $\mathcal{O}(1)$ capacity of the $3$ user interference channel is $(3/2)\log(1+\rho)$ we need to show an innerbound of $(K/2)\log(1+\rho)+\mathcal{O}(1)$. Since our achievable schemes are based on interference alignment and zero forcing, the natural question to ask is whether an interference alignment and zero forcing based scheme can achieve exactly $K/2$ degrees of freedom. The following explanation uses the $K=3$ case to suggest that the answer is negative.

Consider an achievable scheme that uses a $M$ symbol extension of the channel. Now, consider a point $(\alpha_1,\alpha_2,\alpha_3)$ that can be achieved over this extended channel using interference alignment and zero-forcing alone. 
If possible, let the total degrees of freedom over this extended channel be $3M/2$. i.e. $\alpha_1+\alpha_2+\alpha_3 = 3M/2$. It can be argued along the same lines as the converse part of Theorem \ref{theorem:dofach} that $(\alpha_i,\alpha_j)$ is achievable in the $2$ user interference channel for $\forall (i,j) \in \{ (1,2), (2,3), (3,1)\}$.
Therefore
$$ \alpha_1 + \alpha_2 \leq M$$
$$ \alpha_2 + \alpha_3 \leq M$$
$$ \alpha_1 + \alpha_3 \leq M$$

It can be easily seen that the only point $(\alpha_1,\alpha_2,\alpha_3)$ that satisfies the above inequalities and achieves a total of $3M/2$ degrees of freedom is $(\frac{M}{2},\frac{M}{2},\frac{M}{2})$. Therefore, any scheme that achieves a total of $3M/2$ degrees of freedom over the extended channel achieves the point $(\frac{M}{2},\frac{M}{2},\frac{M}{2})$.

We assume that the messages $W_i$ are encoded along $M/2$ independent streams similar to the coding scheme in the proof of Theorem \ref{theorem:dofach} i.e.
$$\barxX^{[i]} = \displaystyle\sum_{m=1}^{M/2}{x}^{[i]}_m \mathbf{v}_m^{[i]} = \barxV^{[i]} \xX^{[i]}$$
Now, at receiver 1, to decode an $M/2$ dimensional signal using zero-forcing, the dimension of the interference has to be at most $M/2$.
i.e.,
\begin{equation}\label{prop1:intalign}\mbox{rank}[\barxH^{[13]} \barxV^{[3]} ~~~ \barxH^{[12]} \barxV^{[2]}] = M/2\end{equation}
Note that since $\barxV^{[2]}$ has $M/2$ linearly independent column vectors and $\barxH^{[12]}$ is full rank with probability 1, $\mbox{rank}(\barxH^{[12]} \barxV^{[2]}) = M/2$. Similarly the dimension of the interference from transmitter $3$ is also equal to $M/2$. Therefore, the two vector spaces on the left hand side of equation (\ref{prop1:intalign}) must have full intersection, i.e 
\begin{eqnarray}
 \label{achint:3by2first}
 \spa(\barxH^{[13]} \barxV^{[3]}) &=& \spa(\barxH^{[12]} \barxV^{[2]}) \\
 \spa(\barxH^{[23]} \barxV^{[3]}) &=& \spa(\barxH^{[21]} \barxV^{[1]}) \mbox{   (At receiver 2) }\\
 \spa(\barxH^{[32]} \barxV^{[2]}) &=& \spa(\barxH^{[31]} \barxV^{[1]}) \mbox {  (At receiver 3) } 
 \label{achint:3by2last}
\end{eqnarray}
where $\spa(\mathbf{A})$ represents the space spanned by the column vectors of matrix $\mathbf{A}$
The above equations imply that 
$$ \spa(\barxH^{[13]} (\barxH^{[23]})^{-1} \barxH^{[21]} \barxV^{[1]}) = \spa(\barxH^{[12]} (\barxH^{[32]})^{-1} \barxH^{[31]} \barxV^{[1]} )$$
$$ \Rightarrow \spa( \barxV^{[1]}) = \spa(\mathbf{T} \barxV^{[1]})$$
where $\mathbf{T} = (\barxH^{[13]})^{-1} \barxH^{[23]} (\barxH^{[21]})^{-1} \barxH^{[12]} (\barxH^{[32]})^{-1}  \barxH^{[31]}$. The above equation implies that there exists at least one eigenvector $\xe$ of $\mathbf{T}$ in $\spa(\barxV^{[1]})$. Note that since all channel matrices are diagonal, the set of eigenvectors of all channel matrices, their inverses and their products are all identical to the set of column vectors of the identity matrix. i.e vectors of the form $[0\mbox{ }0\mbox{ }\ldots\mbox{ }1\mbox{ }\ldots\mbox{ }0]^{T}$.  Therefore $\xe$ is an eigenvector for all channel matrices. Since $\xe$ lies in $\spa(\barxV^{[1]})$, equations (\ref{achint:3by2first})-(\ref{achint:3by2last}) imply that
\begin{eqnarray*}
\xe &\in& \spa(\barxH^{[ij]} \barxV^{[i]}), \forall i,j \in \{1,2,3\}\\
\Rightarrow \xe &\in& \spa(\barxH^{[11]} \barxV^{[1]}) \cap \spa(\barxH^{[12]} \barxV^{[2]})
\end{eqnarray*}
Therefore, at receiver 1, the desired signal $\barxH^{[11]} \barxV^{[1]}$ is \emph{not} linearly independent with the interference $\barxH^{[21]} \barxV^{[2]}$. Therefore, receiver 1 cannot decode $W_1$ completely by merely zero-forcing the interference signal. Evidently, interference alignment in the manner described above cannot achieve exactly $3/2$ degrees of freedom on the $3$ user interference channel with a single antenna at all nodes.

Thus, the degrees of freedom for the $3$ user interference channel with $M=1$ do not automatically lead us to the $\mathcal{O}(1)$ capacity. The possibility that the sum capacity of the $3$ user interference channel with single antennas at all nodes may not be of the form $(3/2)\log(1+\rho) +\mathcal{O}(1)$ is interesting because it suggests that the $3$ user interference channel capacity may not be a straighforward extension of the $2$ user interference channel capacity characterizations.

We explore this interesting aspect of the $3$ user interference channel further in the context of multiple antenna nodes. Our goal is to find out if exactly $3M/2$ degrees of freedom may be achieved with $M$ antennas at each node. As shown by the following theorem, indeed we can achieve exactly $3M/2$ degrees of freedom so that the $\mathcal{O}(1)$ capacity characterization for $M>1$ is indeed related to the degrees of freedom as $\overline{C}(\rho)=(3M/2)\log(1+\rho)$.

\section{Degrees of freedom of the $3$ user interference channel with $M>1$  antennas at each node and constant channel coefficients}\label{section:MIMOint}
The $3$ user MIMO  interference channel is interesting for two reasons. First we wish to show that with multiple antennas we can achieve $3M/2$ degrees of freedom with \emph{constant} channel matrices, i.e., multiple frequency slots are not required. Second, we wish to show that exactly $3M/2$ degrees of freedom are achieved by zero forcing and interference alignment which gives us a lowerbound on sum capacity of $3M/2\log(1+\rho)+\mathcal{O}(1)$. Since the outerbound on sum capacity is also $3M/2\log(1+\rho)+\mathcal{O}(1)$ we have an $\mathcal{O}(1)$ approximation to the capacity of the $3$ user MIMO interference channel with $M>1$ antennas at all nodes.

\begin{theorem}\label{theorem:MIMOint}
In a $3$ user interference channel with $M>1$ antennas at each transmitter and each receiver and constant coefficients, the sum capacity $C(\rho)$ may be characterized as:
\begin{eqnarray}
C(\rho)=(3M/2)\log(1+\rho)+\mathcal{O}(1)
\end{eqnarray}
\end{theorem}
The proof is presented in Appendices \ref{app:proofeven} and \ref{app:proofodd}.

\section{Cognitive Message Sharing on the $3$ user Interference Channel}
Cognitive message sharing refers to a form of cooperation between transmitters and/or receivers where the message of one user is made available non-causally to the transmitter or receiver of another user. Degrees of freedom with cognitive cooperation are considered in \cite{Devroye_Sharif} and \cite{arxiv_dofx2}. It is shown in \cite{Jafar_Shamai} that for the two user interference channel with equal number $M$ of antennas at all nodes there is no gain in degrees of freedom when one user has a cognitive transmitter, a cognitive receiver or both. In all these cases the total degrees of freedom equals $M$.  However the full $2M$ degrees of freedom are obtained if both users have a cognitive transmitter, or both users have a cognitive receiver, or one user has a cognitive transmitter and the other user has a cognitive receiver. In this section we generalize this result to the three user interference channel. 

To generalize the result we introduce some notation. Let $\mathcal{T}_i$ be defined as the set of messages available non-causally at transmitter $i$ and receiver $i$. Let us also define $\mathcal{R}_i$ as the set containing the message intended for receiver $i$ and also the messages non-causally available at receiver $i$. With no cognitive sharing of messages $\mathcal{T}_i=\mathcal{R}_i=\{W_i\}$. Further, let us define $\mathcal{W}_i=\mathcal{T}_i\cup\mathcal{R}_i$.
The following theorem presents the total number of degrees of freedom for some interesting cognitive message sharing scenarios.

\begin{theorem}\label{theorem:cogint}
The total number of degrees of freedom $\eta^\star$ for the $3$ user interference channel under various cognitive message sharing scenarios are determined as follows:
\begin{enumerate}
\item If only one message (e.g. $W_1$) is shared among all nodes the degrees of freedom are unchanged.
\begin{eqnarray}
\mathcal{W}_1=\{W_1\}, \mathcal{W}_2=\{W_1, W_2\}, \mathcal{W}_3=\{W_1,W_3\} \Rightarrow \eta^\star = 3/2.
\end{eqnarray}
Note that this includes all scenarios where message $W_1$ is made available to only the transmitter, only the receiver or both transmitter \emph{and} receiver of users $2$ and $3$. In all these cases, there is no benefit in terms of degrees of freedom.
\item If two messages (e.g. $W_1,W_2$) are shared among all nodes then we have $2$ degrees of freedom.
\begin{eqnarray}
\mathcal{W}_1=\{W_1, W_2\}, \mathcal{W}_2=\{W_1, W_2\}, \mathcal{W}_3=\{W_1,W_2, W_3\} \Rightarrow \eta^\star = 2
\end{eqnarray}
Note that the messages $W_1, W_2$ may be shared through cognitive transmitters, receivers or both.
\item If only one receiver (e.g. receiver 3) is fully cognitive then we have $3/2$ degrees of freedom.
\begin{eqnarray}
\mathcal{W}_1=\{W_1\}, \mathcal{W}_2=\{W_2\}, \mathcal{T}_3=\{W_3\}, \mathcal{R}_3=\{W_1,W_2, W_3\}\Rightarrow \eta^\star = 3/2
\end{eqnarray}
\item If only one transmitter (e.g. transmitter 3) is fully cognitive then we have $2$ degrees of freedom.
\begin{eqnarray}
\mathcal{W}_1=\{W_1\}, \mathcal{W}_2=\{W_2\}, \mathcal{T}_3=\{W_1,W_2, W_3\}, \mathcal{R}_3=\{W_3\}\Rightarrow \eta^\star = 2
\end{eqnarray}
\end{enumerate}
\end{theorem}
The last two cases are significant as they show the distinction between cognitive transmitters and cognitive receivers that was not visible in the two user interference channel studied in \cite{Jafar_Shamai}. In \cite{Jafar_Shamai} it was shown that from a degree of freedom perspective, cognitive transmitters are equivalent to cognitive receivers for the two user interference channel. However, cases $3$ and $4$ above show that cognitive transmitters may be more powerful than cognitive receivers. Intuitively, a cognitive receiver with no message of its own is useless whereas a cognitive transmitter with no message of its own is still useful.

\begin{proof}
\begin{enumerate}
\item Consider the case where $W_1$ is shared with either the transmitter or receiver (or both) of user $2$ and user $3$. Now, with $W_1=\phi$ we have a two user interference channel with no cognitive message sharing which gives us the outerbound $d_2+d_3\leq 1$. For the next outerbound, set $W_2=\phi$ and let the transmitter of user $2$ cooperate with the transmitter of user $1$ as a two antenna transmitter. We now have a two user cognitive interference channel with user $1$ as the primary user (with two transmit antennas and one receive antenna) and user $3$ with the cognitive transmitter, cognitive receiver or both. Using the standard MAC outerbound argument it is easily seen that by reducing the noise at receiver $1$ we must be able to decode both messages $W_1, W_3$ at receiver 1. This gives us the outerbound $d_1+d_3\leq 1$. Similarly, we obtain the outerbound $d_1+d_2\leq 1$. Adding up the three outerbound we have $\eta^\star\leq 3/2$. Since $3/2$ degrees of freedom are achievable even without any cognitive cooperation, we have $\eta^\star=3/2$. Note that this result is easily extended to $K$ users, i.e. with only one message shared among all nodes the degrees of freedom are not increased.
\item For achievability, set $W_3=\phi$ and let transmitter $3$ stay silent. Then we have a cognitive two user interference channel with two shared messages  which has $2$  degrees of freedom as established in \cite{Jafar_Shamai}.  For the converse argument let transmitter $1$ and $2$ cooperate as a two antenna transmitter $T_{12}$ and receiver $1$ and $2$ cooperate as a two antenna receiver $R_{12}$. Then we have a two user cognitive interference channel where the primary user has two transmit and two receive antennas while the cognitive user has a single transmit antenna and a single receive antenna. Once again, the standard MAC outerbound argument is used to show that by reducing noise at the primary receiver, we must be able to decode all messages at the primary receiver. This gives us the outerbound $d_1+d_2+d_3\leq 2$. Since the inner and outerbounds agree, $\eta^\star=2$.
\item Achievability of $3/2$ degrees of freedom is trivial as no cognitive message sharing is required. For the converse, setting $W_1, W_2, W_3$ to $\phi$ one at a time leads to the two user cognitive interference channel with one shared message for which the degrees of freedom are bounded above by $1$. Adding the three outerbounds we conclude that $\eta^\star=3/2$. This result is also easily extended to the $K$ user interference channel.
\item For achievability, set $W_3=\phi$. Then we have a two user interference channel with a cognitive helper (transmitter $3$) who knows both user's messages. Two degrees of freedom are achieved easily on this channel as transmitters $1$ and $3$ cooperate to zero force the transmission of $W_1$ at receiver $2$, while transmitters $2$ and $3$ cooperate to zero force the transmission of $W_2$ at receiver $1$. By eliminating interference at receivers $1$ and $2$, we have two degrees of freedom. The converse follows directly from the converse for part $2$, so that we have $\eta^\star = 2$.
\end{enumerate}
\end{proof}

\section{Conclusion}

We have shown that with perfect channel knowledge the $K$ user interference channel has $K/2$ spatial degrees of freedom.  Conventional wisdom has so far been consistent with the conjecture that distributed interfering systems cannot have more than $1$ degree of freedom and therefore  the best known outerbound $K/2$ has not been considered significant. This pessimistic outlook has for long invited researchers to try to prove that more than $1$ degree of freedom is not possible while ignoring the $K/2$ outerbound. The present result shifts the focus onto the outerbound by proving that it is tight if perfect and global channel knowledge is available. Thus, the present result could guide future research along an optimistic path in the same manner that MIMO technology has shaped our view of the capacity of a wireless channel. There are several promising directions for future work. From a practical perspective it is important to explore to what extent interference alignment can be accomplished with limited channel knowledge. Simpler achievability schemes are another promising avenue of research. For example, the interference alignment scheme based on different propagation delays that we presented in this paper is an exciting possibility as it only requires a careful placing of interfering nodes to satisfy certain delay constraints. 

\appendices

\section{Achievability for  Theorem \ref{theorem:dofach} for arbitrary $K$ }
\label{app:K} 

Let $N=(K-1)(K-2)-1$. We show that $(d_1(n),d_2(n), \ldots d_K(n)) $ lies in the degrees of freedom region of the $K$ user interference channel for any $n \in \mathbb{N}$ where
\begin{eqnarray*}
d_1(n) &=& \frac{(n+1)^N}{(n+1)^N+n^N} \\
d_i(n) &=& \frac{n^N}{(n+1)^N+n^N},~~~~ i=2,3 \ldots K 
\end{eqnarray*}
This  implies that
$$ \max_{(d_1,d_2,\ldots d_K)\in\mathcal{D}} d_1 + d_2 + \cdots d_K \geq \sup_n \frac{(n+1)^N + (K-1) n^N}{(n+1)^N + n^N} = K/2$$

We provide an  achievable scheme to show that $((n+1)^N,n^N,n^N\ldots n^N)$ lies in the degrees of freedom region of an $M_n=(n+1)^N+n^N$ symbol extension of the original channel which automatically implies the desired result.
In the extended channel, the signal vector at the $k^{th}$ user's receiver can be expressed as 
$$ \barxY^{[k]}(t) = \displaystyle\sum_{j=1}^{K} \barxH^{[kj]} \barxX^{[j]}(t) + \barxZ^{[k]}(t)$$ 
where $\barxX^{[j]}$ is an $M_n \times 1$ column vector representing the $M_n$ symbol extension of the transmitted symbol $X^{[k]}$, i.e 
$$\barxX^{[j]}(t) \define \left[ \begin{array}{c} X^{[j]}(1,t) \\X^{[j]}(2,t)\\ \vdots \\ X^{[j]}(M_n,t)) \end{array}\right]$$ Similarly $\barxY^{[k]}$ and $\barxZ^{[k]}$ represent $M_n$ symbol extensions of the $Y^{[k]}$ and $Z^{[k]}$ respectively. 
$\barxH^{[kj]}$ is a diagonal $M_n\times M_n$ matrix representing the $M_n$ symbol extension of the channel i.e
$$ \barxH^{[kj]} \define \left[ \begin{array}{cccc}  H^{[kj]}(1) & 0 & \ldots & 0\\
	0 & H^{[kj]}(2) & \ldots & 0\\
	\vdots & \cdots & \ddots & \vdots\\ 
	0 & 0& \cdots  & H^{[kj]}(M_n) \end{array}\right] $$
Recall that the diagonal elements of $\barxH^{[kj]}$ are drawn independently from a continuous distribution and are therefore distinct with probability $1$.


In a manner similar to the $K=3$ case, message $W_1$ is encoded at transmitter $1$ into $(n+1)^N$ independent streams $x^{[1]}_m(t), m=1,2,\ldots,(n+1)^N$ along vectors $\mathbf{v}^{[1]}_m$ so that $\barxX^{[1]}(t)$ is
$$\barxX^{[1]}(t) = \displaystyle\sum_{m=1}^{(n+1)^N} x^{[1]}_m(t) \mathbf{v}_m^{[1]} = \barxV^{[1]} \xX^{[1]}(t)$$
where $\xX^{[1]}(t)$ is a $(n+1)^N \times 1$ column vector and $\barxV^{[1]}$ is a $ M_n \times (n+1)^N$ dimensional matrix.
Similarly $W_i, i\neq 1$ is encoded into $n^K$ independent streams by transmitter $i$ as 
$$\barxX^{[i]}(t) = \displaystyle\sum_{m=1}^{n^N}{x}^{[i]}_m(t) \mathbf{v}_m^{[i]} = \barxV^{[i]} \xX^{[i]}(t)$$
The received signal at the $i^{th}$ receiver can then be written as
$$ \barxY^{[i]}(t) = \displaystyle\sum_{j=1}^{K}\barxH^{[ij]} \barxV^{[j]}\xX^{[j]}(t) + \barxZ^{[i]}(t)$$

All receivers decode the desired signal by zero-forcing the interference vectors. At receiver 1, to obtain $(n+1)^N$ interference free dimensions corresponding to the desired signal from an $M_n=(n+1)^N+n^N$ dimensional received signal vector ${\barxY}^{[1]}$, the dimension of the interference should be not more than $n^N$. This can be ensured by perfectly aligning the interference from transmitters $2,3 \ldots K$ as follows
\begin{equation}\label{achintK:rx1} \barxH^{[12]} \barxV^{[2]} = \barxH^{[13]} \barxV^{[3]} = \barxH^{[14]} \barxV^{[4]} = \ldots = \barxH^{[1K]} \barxV^{[K]}\end{equation}

At the same time,  receiver $2$ zero-forces the interference from $\barxX^{[i]}, i \neq 2$. To extract $n^N$ interference-free dimensions from a $M_n=(n+1)^N+n^N$ dimensional vector, the dimension of the interference has to be not more than  $(n+1)^N$.

This can be achieved by choosing $\barxV^{[i]}, i \neq 2 $ so that
\begin{equation}
 \label{achintK:rx2}
 \begin{array}{ccc}
\barxH^{[23]}  \barxV^{[3]}  & \prec & \barxH^{[21]} \barxV^{[1]} \\
\barxH^{[24]}  \barxV^{[4]}  & \prec & \barxH^{[21]} \barxV^{[1]} \\
 & \vdots & \\
\barxH^{[2K]}  \barxV^{[K]}  & \prec & \barxH^{[21]} \barxV^{[1]} \\
	\end{array}
	\end{equation}
Notice that the above relations align the interference from $K-2$ transmitters within the interference from transmitter $1$ at receiver $2$.
Similarly, to decode $W_i$ at receiver $i$ when $i\neq 1$  we wish to choose $\barxV^{[i]}$ so that the following $K-2$ relations are satisfied.
\begin{equation} \label{achintK:rxi} \barxH^{[ij]} \barxV^{[j]} \prec \barxH^{[i1]} \barxV^{[1]}, j \notin \{1,i\}  \end{equation}
We now wish to pick vectors $\barxV^{[i]}, i=1,2 \ldots K $ so that equations (\ref{achintK:rx1}), (\ref{achintK:rx2}) and (\ref{achintK:rxi}) are satisfied.
Since channel matrices $\barxH^{[ij]}$ have a full rank of $M_n$ almost surely,  equations (\ref{achintK:rx1}), (\ref{achintK:rx2}) and (\ref{achintK:rxi}) can be equivalently expressed as
\begin{eqnarray}
 \label{achintK:rxfirst}
 \barxV^{[j]}  =  \mathbf{S}^{[j]} \mathbf{B} & j=2,3,4 \ldots K& \mbox{ At receiver 1}\\
	 \left.
	 \begin{array}{ccc}
 \mathbf{T}^{[2]}_3 \mathbf{B} = \mathbf{B}  & \prec &  \barxV^{[1]} \\
 \mathbf{T}^{[2]}_4 \mathbf{B} & \prec & \barxV^{[1]} \\
 & \vdots & \\
 \mathbf{T}^{[2]}_{K} \mathbf{B} & \prec & \barxV^{[1]} 
      \end{array}
		 \right\} &  \textrm{At receiver 2 } &   \label{achintK:rxmiddle}\\
	 \left.
	 \begin{array}{ccc}
\mathbf{T}^{[i]}_2  \mathbf{B}  & \prec &  \barxV^{[1]} \\
 \mathbf{T}^{[i]}_3 \mathbf{B} & \prec & \barxV^{[1]} \\
 & \vdots & \\
 \mathbf{T}^{[i]}_{i-1} \mathbf{B} & \prec & \barxV^{[1]} \\
 \mathbf{T}^{[i]}_{i+1} \mathbf{B} & \prec & \barxV^{[1]} \\
 & \vdots & \\
 \mathbf{T}^{[i]}_{K} \mathbf{B} & \prec & \barxV^{[1]} 
      \end{array}
		 \right\} &  \textrm{At receiver i where } i=3 \ldots K &  
 \label{achintK:rxlast}
\end{eqnarray}

where 
\begin{eqnarray}
\mathbf{B} = (\barxH^{[21]})^{-1} \barxH^{[23]} \barxV^{[3]}\\
\mathbf{S}^{[j]}  =  (\barxH^{[1j]})^{-1}\barxH^{[13]} (\barxH^{[23]})^{-1} \barxH^{[21]}, & j=2,3, \ldots K \\
\mathbf{T}^{[i]}_{j}  =  (\barxH^{[i1]})^{-1} \barxH^{[ij]} \mathbf{S}^{[j]} & i,j=2,3\ldots K , j\neq i 
 \label{achintK:T_last}
\end{eqnarray}
Note that $\mathbf{T}^{[2]}_3 = \mathbf{I}$, the $M_{n} \times M_{n}$ identity matrix.
We now choose $\barxV^{[1]}$ and $\mathbf{B}$ so that they satisfy the $(K-2)(K-1)=N+1$ relations in (\ref{achintK:rxmiddle})-(\ref{achintK:rxlast}) and then use equations in (\ref{achintK:rxfirst}) to determine $\barxV^{[2]}, \barxV^{[3]} \ldots \barxV^{[K]}$. Thus,  our goal is to find matrices $\barxV^{[1]}$ and $\mathbf{B}$ so that 
$$ \mathbf{T}^{[i]}_j \mathbf{B} \prec \barxV^{[1]}$$ for all $i,j = \{2,3 \ldots K\},i\neq j$.

Let $\mathbf{w}$ be the $M_n\times 1$ column vector 
$$\mathbf{w} = \left[ \begin{array}{c} 1 \\ 1 \\ \vdots \\ 1\end{array} \right]$$

We need to choose $n^{(K-1)(K-2)-1} = n^N$ column vectors for $\mathbf{B}$. The sets of column vectors of $\mathbf{B}$ and $\barxV^{[1]}$ are chosen to be equal to the sets $B$ and $\bar{V}^{[1]}$ where 

$$ B =  \left\{ \bigg(\prod_{m,k \in \{2,3, \ldots K\}, m\neq k, (m,k)\neq(2,3)}\big( \mathbf{T}^{[m]}_k \big)^{\alpha_{mk}} \bigg) \xw: \forall \alpha_{mk} \in \{0,1,2 \ldots n-1 \}\right\}$$
$$ \bar{V}^{[1]} = \left\{ \bigg(\prod_{m,k \in \{2,3, \ldots K\}, m\neq k,(m,k)\neq(2,3)}\big( \mathbf{T}^{[m]}_k \big)^{\alpha_{mk}}\bigg) \xw: \forall \alpha_{mk} \in \{0,1,2 \ldots n\}\right\}$$

For example, if $K=3$ we get $N=1$. $\mathbf{B}$ and $\barxV^{[1]}$ are chosen as 
\begin{eqnarray*} \mathbf{B} &=& \left[ \xw ~~~\mathbf{T}^{[3]}_2 \xw~~~\ldots~~~(\mathbf{T}^{[3]}_2)^{n-1}\xw\right] \\
\barxV^{[1]} &=& \left[ \xw ~~~\mathbf{T}^{[3]}_2 \xw~~~\ldots~~~(\mathbf{T}^{[3]}_2)^{n} \xw \right]
\end{eqnarray*}
To clarify the notation further, consider the case where $K=4$.  Assuming $n=1$, $B$ consists of exactly one element i.e  $B=\{\xw\}$. The set $\bar{V}^{[1]}$ consists of all $2^N=2^5=32$ column vectors of the form\\
$(\mathbf{T}^{[2]}_4)^{\alpha_{24}} (\mathbf{T}^{[3]}_2)^{\alpha_{32}}(\mathbf{T}^{[3]}_4)^ {\alpha_{24}}(\mathbf{T}^{[4]}_3)^{\alpha_{43}} (\mathbf{T}^{[4]}_2)^{\alpha_{42}} \xw$\\
where all $\alpha_{24}, \alpha_{32}, \alpha_{34}, \alpha_{42}, \alpha_{43}$ take values $0,1$. $B$ and $\bar{V}^{[1]}$ can be verified to have $n^N$ and $(n+1)^{N}$ elements respectively.

$\barxV^{[i]}, i=2, 3\ldots K$ are chosen using equations (\ref{achintK:rxfirst}).
Clearly, for $(i,j)=(2,3)$, 
	$$ \mathbf{T}^{[i]}_j \mathbf{B} = \mathbf{B} \prec \barxV^{[1]}$$
Now, for $i \neq j, i,j=2 \ldots K, (i,j) \neq (2,3) $
	\begin{eqnarray*}  
\mathbf{T}^{[i]}_j B &=& \Bigg\{ \bigg(\prod_{m,k \in \{2,3, \ldots N\}, m\neq k, (m,k) \neq (2,3)}\big( \mathbf{T}^{[m]}_k \big)^{\alpha_{mk}} \bigg) \xw: \\
& & \forall (m,k) \neq (i,j), \alpha_{mk} \in \{0,1,2 \ldots n-1\}, \alpha_{ij} \in \{1,2, \ldots n\} \Bigg\}\\
	 \Rightarrow \mathbf{T}^{[i]}_{j} B &\in& \bar{V}^{[1]} \\
	 \Rightarrow \mathbf{T}^{[i]}_{j} \mathbf{B} &\prec& \mathbf{\bar{V}^{[1]}} \\
	 \end{eqnarray*}
Thus, the interference alignment equations (\ref{achintK:rxfirst})-(\ref{achintK:rxlast}) are satisfied.

Through interference alignment, we have now ensured that the dimension of the interference is small enough. We now need to verify that the components of the desired signal are linearly independent of the components of the interference so that the signal stream can be completely decoded by zero-forcing the interference.
Consider the received signal vectors at Receiver $1$. The desired signal arrives along the $(n+1)^N$ vectors $\barxH^{[11]} \barxV^{[1]}$. As enforced by equations (\ref{achintK:rxfirst}), the interference vectors from transmitters $3,4 \ldots K$ are perfectly aligned with the interference from transmitter $2$ and therefore, all interference arrives along the $n^N$ vectors $\barxH^{[12]} \barxV^{[2]}$. In order to prove that there are $(n+1)^N$ interference free dimensions it suffices to show that the columns of the square, $M_n\times M_n$ dimensional matrix
\begin{eqnarray}
\left[\barxH^{[11]} \barxV^{[1]}~~~~\barxH^{[12]} \barxV^{[2]}\right]
\end{eqnarray}
are linearly independent almost surely. Multiplying the above $M_n \times M_n$ matrix with $(\barxH^{[11]})^{-1}$ and substituting for $\barxV^{[1]}$ and $\barxV^{[2]}$, we get a matrix whose $l$th row has entries of the forms 
$$\prod_{(m,k) \in \{2,3  \ldots K\}, m \neq k, (m,k) \neq (2,3)} (\lambda_l^{mk})^{\alpha_{mk}}$$ and $$d_l \prod_{(m,k) \in \{2,3  \ldots K\}, m \neq k, (m,k) \neq (2,3)} (\lambda_l^{mk})^{\beta_{mk}}$$ 
where $\alpha_{mk} \in \{0,1,\ldots n-1\}$ and $\beta_{mk} \in \{ 0,1,\ldots n\}$ and $\lambda_l^{mk},d_l$  are drawn independently from a continuous distribution.
The same iterative argument as in section \ref{proof:dofach} can be used. i.e. expanding the corresponding determinant along the first row, the linear independence condition boils down to one of the following occurring with non-zero probability
\begin{enumerate}
\item $d_1$ being equal to one of the roots of a linear equation
\item The coefficients of the above mentioned linear equation being equal to zero
\end{enumerate}
Thus the iterative argument can be extended here, stripping the last row and last column at each iteration and the linear independence condition can be shown to be equivalent to the linear independence of a $n^N \times n^N$ matrix whose rows are of the form
$\prod_{(m,k) \in \{2,3  \ldots K\}, m \neq k} (\lambda_l^{mk})^{\alpha_{mk}}$
where $\alpha_{pq} \in \{0,1,\ldots n-1\}$.
Note that this matrix is a more general version of the Vandermonde matrix obtained in section \ref{proof:dofach}. So the argument for the $K=3$ case does not extend here. However, the iterative procedure which eliminated the last row and the last column at each iteration, can be continued. For example, expanding the determinant along the first row, the singularity condition simplifies to one of 
\begin{enumerate}
\item $\lambda_l^{mk}$ being equal to one of the roots of a finite degree polynomial
\item The coefficients of the above mentioned polynomial being equal to zero
\end{enumerate}
Since the probability of condition $1$ occurring is $0$, condition $2$ must occur with non-zero probability. Condition $2$ leads to a polynomial in another random variable $\lambda^{pq}_l$ and thus the iterative procedure can be continued until the linear independence condition is shown to be equivalent almost surely to a $1 \times 1$ matrix being equal to $0$. Assuming, without loss of generality, that we placed the $\xw$ in the first row (this corresponds to the term $\alpha_{mk} = 0, \forall (m,k)$), the linear independence condition boils down to the condition that $ 1 = 0 $ with non-zero probability - an obvious contradiction. Thus the matrix
$$\left[\barxH^{[11]} \barxV^{[1]}~~~~\barxH^{[12]} \barxV^{[2]}\right]$$
can be shown to be non-singular with probability $1$.

Similarly, the desired signal can be chosen to be linearly independent of the interference at all other receivers almost surely. Thus $(\frac{(n+1)^N}{(n+1)^N+n^N},\frac{n^N}{(n+1)^N+n^N},\cdots,\frac{n^N}{(n+1)^N+n^N})$ lies in the degrees of freedom region of the $K$ user interference channel and therefore, the $K$ user interference channel has $K/2$ degrees of freedom.

\section{Proof of Theorem \ref{theorem:MIMOint} for $M$ even}\label{app:proofeven}
\begin{proof}
The outerbound is straightforward as before.
To prove achievability we first consider the case when $M$ is even. Through an achievable scheme, we show that there are $M/2$ non-interfering paths between transmitter $i$ and receiver $i$ for each $i=1,2,3$ resulting in a total of $3M/2$ paths in the network.

Transmitter $i$ transmits message $W_{i}$ for receiver $i$ using $M/2$ independently encoded streams over vectors $\mathbf{v}^{[i]}$ i.e
$$\xX^{[i]}(t) = \displaystyle\sum_{m=1}^{M/2}x^{[i]}_m(t) \mathbf{v}_m^{[i]} = \xV^{[i]} \xX^{i}(t), i=1,2,3$$
The signal received at receiver $i$ can be written as
$$ \xY^{[i]}(t) = \xH^{[i1]}\xV^{[1]} \xX^1(t) + \xH^{[i2]} \xV^{[2]} \xX^2(t) + \xH^{[i3]}\xV^{[3]} \xX^3(t) + \xZ_i(t)$$
All receivers cancel the interference by zero-forcing and then decode the desired message. To decode the $M/2$ streams along the column vectors of $\xV^{[i]}$ from the $M$ components of the received vector, the dimension of the interference has to be less than or equal to $M/2$. 
The following three interference alignment equations ensure that the dimension of the interference is equal to $M/2$  at all the receivers. 


\begin{eqnarray}
\label{achintM:rxfirst}
  \spa(\xH^{[12]} \xV^{[2]}) &=& \spa(\xH^{[13]} \xV^{[3]})\\
  \xH^{[21]} \xV^{[1]} &=& \xH^{[23]} \xV^{[3]}\\
  \xH^{[31]}\xV^{[1]} &=& \xH^{[32]} \xV^{[2]}
\label{achintM:rxlast}
\end{eqnarray}
where $\spa(\mathbf{A})$ represents the vector space spanned by the column vectors of matrix $\mathbf{A}$
We now wish to choose $\xV^{[i]}, i=1,2,3$ so that the above equations are satisfied. Since $\xH^{[ij]},i,j \in \{1,2,3\}$ have a full rank of $M$ almost surely, the above equations can be equivalently represented as

\begin{eqnarray}
\label{achintM:rxfirst_eq}
  \spa(\xV^{[1]}) &=& \spa(\mathbf{E} \xV^{[1]})\\
  \xV^{[2]} &=& \mathbf{F} \xV^{[1]}\\
  \xV^{[3]} &=& \mathbf{G} \xV^{[1]}
\label{achintM:rxlast_eq}
\end{eqnarray}
where 
\begin{eqnarray*}
\mathbf{E} &=&(\xH^{[31]})^{-1} \xH^{[32]} (\xH^{[12]})^{-1} \xH^{[13]} (\xH^{[23]})^{-1} \xH^{[21]} \\
\mathbf{F} &=&(\xH^{[32]})^{-1} \xH^{[31]}\\
\mathbf{G} &=&(\xH^{[23]})^{-1} \xH^{[21]}
\end{eqnarray*}

Let $\xe_1, \xe_2, \ldots \xe_M$ be the $M$ eigenvectors of $\mathbf{E}$. Then we set $\xV_1$ to be
$$ \xV^{[1]} = \begin{array}{ccc} [\xe_1 & \ldots & \xe_{(M/2)}]\end{array} $$
Then $\xV^{[2]}$ and $\xV^{[3]}$ are found using equations (\ref{achintM:rxfirst_eq})-(\ref{achintM:rxlast_eq}). Clearly, $\xV^{[i]}, i=1,2,3$ satisfy the desired interference alignment equations (\ref{achintM:rxfirst})-(\ref{achintM:rxlast}).
Now, to decode the message using zero-forcing, we need the desired signal to be linearly independent of the interference at the receivers. For example, at receiver $1$, we need the columns of $\xH^{[11]} \xV^{[1]}$ to be linearly independendent with the columns of $\xH^{[21]} \xV^{[2]}$ almost surely. i.e we need the matrix below to be of full rank almost surely
$$ \left[\begin{array}{cc} \xH^{[11]} \xV^{[1]} & \xH^{[12]} \xV^{[2]} \end{array}\right] $$
Substituting values for $\xV^{[1]}$ and $\xV^{[2]}$ in the above matrix, and multiplying by full rank matrix $(\xH^{[11]})^{-1}$, the linear independence condition is equivalent to the condition that the column vectors of
$$  \left[ \begin{array}{ccccc} \xe_1 & \xe_2 & \ldots \xe_{(M/2)} & \mathbf{K} \xe_1 & \ldots \mathbf{K} \xe_{(M/2)} \end{array} \right]$$
are linearly independent almost surely, where $\mathbf{K} = (\xH^{[11]})^{-1}\xH^{[12]} \mathbf{F}$.  

This is easily seen to be true because $\mathbf{K}$ is a random (full rank) linear transformation. To get an intuitive understanding of the linear independence condition, consider the case of $M=2$. Let $\mathcal{L}$ represent the line along which lies the first eigenvector of the random $2\times2$ matrix $\mathbf{E}$. The probability of a random rotation (and scaling) $\mathbf{K}$ of $\mathcal{L}$ being collinear with $\mathcal{L}$ is zero. 


Using a similar argument, we can show that matrices 
\begin{eqnarray*} \left[\begin{array}{cc} \xH^{[22]} \xV^{[2]} & \xH^{[21]} \xV^{[1]} \end{array}\right] &\mbox{and}&
 \left[\begin{array}{cc} \xH^{[33]} \xV^{[3]} & \xH^{[31]} \xV^{[1]} \end{array}\right] \end{eqnarray*}  have a full rank of $M$ almost surely and therefore receivers $2$ and $3$ can decode the $M/2$ streams of $\xV^{[2]}$ and $\xV^{[3]}$ using zero-forcing. Thus, a total $3M/2$ interference free transmissions per channel-use are achievable with probability $1$ and the proof is complete.
\end{proof}
\section{Proof of Theorem \ref{theorem:MIMOint} for $M$ odd}\label{app:proofodd}
\begin{proof}
Consider a two time-slot symbol extension of the channel, with the same chanel coefficients over the two symbols. It can be expressed as 
\begin{eqnarray*} \barxY^{[k]} = \barxH^{[k1]} \barxX^{[1]} + \barxH^{[k2]} \barxX^{[2]} + \barxH^{[k3]}\xX^{[3]} + \barxZ^{[k]} & i=1,2,3 \end{eqnarray*}
where $\barxX^{[i]}$ is a $2M \times 1$ vector that represents the two symbol extension of the transmitted $M\times1$ symbol symbol $\xX^{[k]}$, i.e 
$$ \barxX^{[k]}(t) \define \left[ \begin{array}{c} \xX^{[k]}(1,2t+1) \\ \xX^{[k]}(1,2t+2) \end{array} \right]$$ where $\xX^{[k]}(t)$ is an $M\times1$ vector representing the vector transmitted at time slot $t$ by transmitter $k$. Similarly $\barxY^{[k]}$ and $\barxZ^{[k]}$ represent the two symbol extensions of the the received symbol $\xY^{[k]}$ and the noise vector $\xZ^{[k]}$ respectively at receiver $i$.
$\barxH^{[ij]}$ is a $2M\times2M$ block diagonal matrix representing the extension of the channel.
$$ \barxH^{[ij]} \define \left[ \begin{array}{cc}  \xH^{[ij]}(1) & 0 \\
	0 & \xH^{[ij]}(1) \end{array}\right] $$
We will now show $(M,M,M)$ lies in the degrees of freedom region of this extended channel channel with an achievable scheme, implying that that a total of $3M/2$ degrees of freedom are achievable over the original channel.
Transmitter $k$ transmits message $W_{i}$ for receiver $i$ using $M$ independently encoded streams over vectors $\mathbf{v}^{[k]}$ i.e
$$\barxX^{[k]} = \displaystyle\sum_{m=1}^{M}x^{[k]}_m \mathbf{v}_m^{[k]} = \barxV^{[k]} \xX^{[k]}$$
where $\barxV^{[k]}$ is a $2M \times M$ matrix and $\barxX^[k]$ is a $M\times 1$ vector representing $M$ independent streams.
The following three interference alignment equations ensure that the dimension of the interference is equal to $M$  at receivers $1$,$2$ and $3$.
\begin{eqnarray}
\label{achintM:rx123odd}
  \mbox{rank}[\barxH^{[21]} \barxV^{[2]}] &=&  \mbox{rank}[\barxH^{[31]} \barxV^{[3]}]\\
 \barxH^{[12]} \barxV^{[1]} &=& \barxH^{[32]} \barxV^{[3]}\\
 \barxH^{[13]} \barxV^{[1]} &=& \barxH^{[23]} \barxV^{[2]}
\end{eqnarray}
The above equations imply that 
\begin{eqnarray}
\label{achintM:rxfirstodd}
  \spa(\barxV^{[1]}) &=& \spa(\mathbf{\bar{E}} \barxV^{[1]})\\
  \barxV^{[2]} &=& \mathbf{\bar{F}} \barxV^{[1]}\\
  \barxV^{[3]} &=& \mathbf{\bar{G}} \barxV^{[1]}
\label{achintM:rxlastodd}
\end{eqnarray}
where 
\begin{eqnarray*}
\mathbf{E} &=&(\xH^{[13]})^{-1} \xH^{[23]} (\xH^{[21]})^{-1} \xH^{[31]} (\xH^{[32]})^{-1} \xH^{[12]} \\
\mathbf{F} &=&(\xH^{[13]})^{-1} \xH^{[23]}\\
\mathbf{G} &=&(\xH^{[12]})^{-1} \xH^{[32]}
\end{eqnarray*}
and $\mathbf{\bar{E}}$, $\mathbf{\bar{F}}$ and $\mathbf{\bar{G}}$  are $2M\times 2M$ block-diagonal matrices representing the $2M$ symbol extension of $\mathbf{E}$, $\mathbf{F}$ and $\mathbf{G}$ respectively.
Let $\xe_1, \xe_2,\ldots,\xe_M,$ be the eigen vectors of $\mathbf{E}$. Then, we pick $\barxV^{[1]}$ to be 
\begin{equation} 
\label{achintM:oddV1} 
\barxV^{[1]} = \left[ \begin{array}{cccccc} 
\xe_1 & 0 & \xe_3 & \ldots & 0 & \xe_{M} \\ 
0 & \xe_2 & 0 & \ldots & \xe_{M-1} & \xe_{M} \end{array} \right] \end{equation}
As in the even $M$ case, $\barxV^{[2]}$ and $\barxV^{[3]}$ are then determined by using equations (\ref{achintM:rxfirstodd})-(\ref{achintM:rxlastodd}). 

Now, we need the desired signal to be linearly independent of the interference at all the receivers. At receiver $1$, the desired linear independence condition boils down to 
\begin{eqnarray*}
 \spa(\barxV^{[1]}) \cap \spa(\mathbf{\bar{K}} \barxV^{[1]}) &=& \{0\}
\end{eqnarray*}
where $\mathbf{K} = (\xH^{11})^{-1}\xH^{[21]} (\mathbf{{F}})^{-1}$ and $\mathbf{\bar{K}}$ is the two-symbol diagonal extension of $\mathbf{K}$. Notice that $\mathbf{K}$ is an $M\times M$ matrix. The linear independence condition is equivalent to saying that all the columns of the following $2M \times 2M$ matrix are independent.
\begin{equation} \label{achintM:linindodd} \left[ \begin{array}{cccccccccccc}
\xe_1 & 0 & \xe_3 & \ldots & 0 & \xe_{M}& \mathbf{K} \xe_1 & 0 & \mathbf{K} \xe_3 & \ldots & 0 & \mathbf{K} \xe_{M}\\
 0 & \xe_2 & 0 & \ldots &  \xe_{M-1} & \xe_{M}  & 0 &  \mathbf{K} \xe_2 & 0 & \ldots &  \mathbf{K} \xe_{M-1} & \mathbf{K}  \xe_{M}  \\
 \end{array} \right]\end{equation}
We now argue that the probability of the columns of the above matrix being linearly dependent is zero. Let $\mathbf{c}_i, i=1,2 \ldots 2M$ denote the columns of the above matrix. Suppose the columns $\mathbf{c}_i$ are linearly dependent, then
\begin{eqnarray*} \exists \alpha_i & \mbox{s.t} & \displaystyle\sum_{i=1}^{2M} \alpha_i \mathbf{c}_i = 0 \end{eqnarray*}
Let
\begin{eqnarray*}
 \mathbf{P} &=& \{\xe_1, \xe_3 \ldots \xe_{M-2}, \mathbf{K} \xe_1, \ldots \mathbf{K} \xe_{M-2}\} \\ 
 \mathbf{Q} &=& \{\xe_2, \xe_4 \ldots \xe_{M-1}, \mathbf{K} \xe_2, \ldots \mathbf{K} \xe_{M-1}\}
\end{eqnarray*}
Now, there are two possibilities
\begin{enumerate}
\item $\alpha_{M} = \alpha_{2M} = 0$. This implies that either one of the following sets of vectors is linearly dependent.
Note that both sets are can be expressed as the union of
\begin{enumerate}
\item A set of $\lfloor(M/2)\rfloor$ eigen vectors of $\mathbf{E}$ 
\item A random transformation $\mathbf{K}$ of this set.
\end{enumerate}
An argument along the same lines as the even $M$ case leads to the conclusion that the probability of the union of the two sets listed above being linearly dependent in a $M$ dimensional space is zero.
\item $\alpha_{2M} \neq 0$ or $\alpha_{M} \neq 0$
This implies that
$$\alpha_{M} \xe_M + \alpha_{2M}\mathbf{K} \xe_{M} \in \spa( \mathbf{P}) \cap \spa(\mathbf{Q}) $$
$$ \Rightarrow \spa(\{ \mathbf{K} \xe_{M}, \xe_{M}\})  \cap \spa( \mathbf{P}) \cap \spa(\mathbf{Q}) \neq \{0\}$$
Also
$$\mbox{rank}(\spa({\mathbf{P}}) \cup \spa(\mathbf{Q})) = \mbox{rank}(\mathbf{P}) + \mbox{rank} (\mathbf{Q}) - \mbox{rank}( \mathbf{P} \cap \mathbf{Q} )$$
$$\Rightarrow \mbox{rank}( \mathbf{P} \cap \mathbf{Q} ) = 2M-2-\mbox{rank}(\spa({\mathbf{P}}) \cup \spa(\mathbf{Q})) $$
Note that $\mathbf{P}$ and $\mathbf{Q}$ are $M-1$ dimensional spaces. (The case where their dimensions are less than $M-1$ is handled in the first part). Also, $\mathbf{P}$ and $\mathbf{Q}$ are drawn from completely different set of vectors. Therefore, the union of $\mathbf{P},\mathbf{Q}$ has a rank of $M$ almost surely.  Equivalently $\spa(\mathbf{P}) \cap \spa(\mathbf{Q})$ has a dimension of $M-2$ almost surely. Since the set $\{ \xe_M, \mathbf{K} \xe_M\}$ is drawn from an eigen vector $\xe_M$ that does not exist in either $\mathbf{P}$ or $\mathbf{Q}$, the probability of the $2$ dimensional space $\spa(\{\xe,\mathbf{K} \xe_M\})$ intersecting with the $M-2$ dimensional space $\mathbf{P} \cap \mathbf{Q}$ is zero.
For example, if $M=3$, let $L$ indicate the line formed by the intersection of the the two planes $\spa(\{\xe_1,\mathbf{K} \xe_1\})$ and $\spa(\{\xe_2, \mathbf{K} \xe_2\})$. The probability that line $L$ lies in the plane formed by $ \spa(\{\xe_{3}, \mathbf{K} \xe_3\})$.
Thus, the probability that the desired signal lies in the span of the interference is zero at receiver $1$. Similarly, it can be argued that the desired signal is independent of the interference at receivers $2$ and $3$ almost surely. Therefore $(M,M,M)$ is achievable over the two-symbol extended channel. Thus $3M/2$ degrees of freedom are achievable over the $3$ user interference channel with $M$ antenna at each transmitting and receiving node.
\end{enumerate}
\end{proof}
\bibliographystyle{ieeetr}
\bibliography{Thesis}
\end{document}